\documentclass[letterpaper]{article} 
\usepackage{aaai2027}  
\usepackage[hyphens]{url}  
\usepackage{graphicx} 
\usepackage{natbib}  
\usepackage{caption} 
\usepackage{algorithm}
\usepackage{algorithmic}

\usepackage{newfloat}
\usepackage{listings}
\DeclareCaptionStyle{ruled}{labelfont=normalfont,labelsep=colon,strut=off} 
\floatstyle{ruled}
\newfloat{listing}{tb}{lst}{}
\floatname{listing}{Listing}

\usepackage{booktabs}
\usepackage{amsfonts}       
\usepackage{amsthm}         
\newtheorem{theorem}{Theorem}
\newtheorem{proposition}{Proposition}
\newtheorem{corollary}{Corollary}

\usepackage{nicefrac}       
\usepackage{microtype}      
\usepackage{graphicx}
\usepackage{array}
\usepackage[table]{xcolor}
\usepackage{multirow}
\usepackage{caption}
\usepackage[most]{tcolorbox}
\usepackage{enumitem}
\usepackage{listings}
\title{Strategic Calibration in LLM Multi-Agent Systems: Audited Pivot Scoring for Truthful Elicitation}

\author{
    Jinhao Zhou\textsuperscript{\rm 1},
    Zhaohan Zhang\textsuperscript{\rm 2},
    Xijian Xu\textsuperscript{\rm 1},
    Jun Wu\textsuperscript{\rm 3}
}

\affiliations{
    \textsuperscript{\rm 1}Graduate School of Information, Production and Systems,
    Waseda University, Fukuoka, Japan\\
    \textsuperscript{\rm 2}School of Electronic Engineering and Computer Science,
    Queen Mary University of London, London, U.K.\\
    \textsuperscript{\rm 3}School of Computer Science,
    Shanghai Jiao Tong University, Shanghai, China\\
    jinhao@suou.waseda.jp,
    zhaohan.zhang@qmul.ac.uk,
    xijian.xu@asagi.waseda.jp,
    jun.wu@ieee.org
}

\begin{document}

\maketitle

\begin{abstract}
In large language model (LLM)-based multi-agent systems,
reported confidence increasingly determines the allocation of
scarce downstream attention, computation, and decision influence. This makes reporting
strategic: an agent may overstate the relevance of its
contribution even when its underlying belief is well
calibrated. Existing calibration methods largely assume passive reporting and therefore do not address strategic incentives created when confidence determines downstream allocation.
We formulate strategic calibration as a reporting-layer incentive problem and propose the Audited Pivot Scoring Mechanism
(APSM). APSM offsets platform-controlled allocation benefits,
uses inverse-propensity-weighted sparse audits, charges the
Clarke pivot externality, and applies a strictly proper score.
Under the stated model assumptions, truthful reporting is the unique
expected-utility best response. We also characterize approximate execution and the
audit--variance trade-off. Calculator-free experiments show
that LLM agents inflate confidence under Top-$K$ incentives.
With a shared agent-local utility interface, LLM agents reduce
distortion under both Compensated Proper Scoring (CPS) and
APSM. CPS achieves higher exact allocation match, whereas
APSM yields lower Brier loss and numerically fewer cutoff crossings.
Experiments further show that APSM enlarges the utility
margin against allocation-changing deviations by
$1.93\times$ and improves robustness under approximate
execution. These results support incentive-aware elicitation as
a distinct layer for reliable multi-agent coordination.
\end{abstract}

\section{Introduction}
\label{sec:introduction}

In large language model (LLM)-based multi-agent systems,
confidence is increasingly an allocative signal rather than
merely a description of uncertainty. Multiple agents may
contribute evidence, plans, tool outputs, or candidate answers,
after which a coordinator decides which contributions receive
limited downstream attention, computation, or decision
weight~\cite{hong2024metagpt,wu2023autogen,
du2023improving}. Confidence provides a natural signal for
making such decisions. Once the same report also determines
an agent's exposure or influence, however, confidence
reporting becomes a strategic action. An agent may benefit
from overstating its contribution to obtain
greater downstream influence, even when the resulting
allocation degrades the coordinator's decision.

This failure mode is distinct from intrinsic model
miscalibration. Existing approaches improve confidence
quality through post-hoc calibration, prompted elicitation, or multi-agent deliberation
\cite{guo2017calibration,kadavath2022language,
lin2022teaching,yang2024confidence,pandey2026refine}.
These methods largely treat a reported probability as a
passive prediction rather than an action that affects the
predictor's utility. A system may therefore produce meaningful
agent-local beliefs while still eliciting distorted reports when
confidence controls selection. Recent evidence that LLM
agents adapt their behavior to persuasive or utility-seeking
objectives further motivates explicitly modeling reporting
incentives~\cite{su2025ai,agarwal2025persuasion}. A strictly proper score does not by itself eliminate this
problem when the allocation benefit remains uncompensated.
Under capacity-constrained selection, a small increase in a
report can yield a discrete gain in downstream influence,
whereas the corresponding proper-score loss vanishes
smoothly with the reporting deviation. Near hard allocation boundaries, strategic exaggeration may remain
profitable, so calibration must account for report-dependent allocation
incentives.

\begin{figure}[t]
    \centering
    \includegraphics[width=0.8\columnwidth]
    {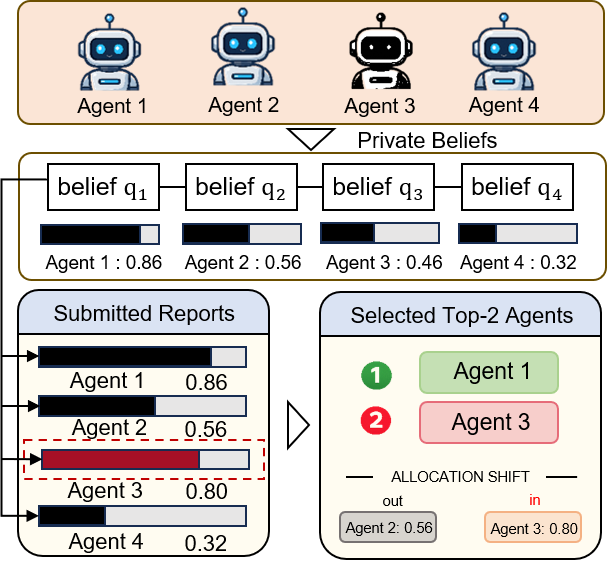}
    \caption{Strategic misreporting alters allocation:
    Agent~3 inflates its report and displaces Agent~2.}
    \label{fig:background}
\end{figure}

We formulate strategic calibration as a reporting-layer
incentive problem and propose the Audited Pivot Scoring
Mechanism (APSM). APSM compensates platform-controlled
allocation benefits, uses inverse-propensity-weighted sparse
audits, charges the Clarke pivot externality, and applies a
strictly proper score. Its Groves term enlarges the utility margin
against allocation-changing deviations, improving robustness
to approximate utility evaluation. We establish strict expected
truthfulness and characterize approximate execution, sparse
auditing, and payoff misspecification. Experiments with
calculator-free LLMs, utility-assisted reporting, and controlled
replay in two workflows validate both the reporting failure and
APSM's allocation-sensitive advantage. Our main contributions are:

\begin{itemize}
    \item We formulate strategic calibration as a
    reporting-layer incentive problem in LLM multi-agent
    systems. We show that uncompensated proper scoring
    cannot generally eliminate threshold manipulation when
    confidence determines scarce Top-$K$ allocation.

    \item We propose APSM, a new allocation-sensitive elicitation mechanism that internalizes report-induced allocation
externalities under sparse outcome verification. We establish
strict expected truthfulness and show that its Groves structure
enlarges the utility margin against allocation-changing
deviations beyond CPS.

    \item We separate behavioral and mechanism evidence
    through calculator-free LLM reporting, LLM-controlled
    utility-assisted execution, and controlled replay in two
    task-grounded workflows. The results characterize the CPS--APSM execution trade-off and show that APSM improves robustness under approximate execution even after controlling for utility scale.
\end{itemize}

\section{Method}
\label{sec:method}

We formulate strategic calibration as a reporting-layer
mechanism-design problem. A calibration or deliberation
procedure may produce a meaningful agent-local relevance
belief, yet the submitted confidence can still be distorted
when it determines scarce downstream allocation. Our
objective is therefore not to recalibrate the underlying belief,
but to make faithful reporting incentive-compatible.

We propose APSM, an audited allocation-sensitive elicitation
mechanism that internalizes report-induced allocation
externalities under sparse outcome verification. APSM offsets
platform-controlled allocation benefits, assigns verified
outcome value, prices the opportunity cost imposed on other
agents, and uniquely identifies the relevance probability. We
describe a generic round $t$; repeated deployment applies the
mechanism independently across rounds.

\subsection{Problem Setup}
\label{subsec:problem-setup}

At round $t$, each of $N$ agents produces a candidate
$c_{i,t}$ before confidence elicitation. Candidate $i$ has a
binary verifiable outcome $y_{i,t}\in\{0,1\}$. Based on the
query $x_t$, the fixed candidate, and its local information
$\mathcal F_{i,t}$, agent $i$ forms the relevance belief
$q_{i,t}$ but submits an observable confidence report
$r_{i,t}$:
\begin{equation}
q_{i,t}
:=
\Pr\!\left(
y_{i,t}=1
\mid
x_t,c_{i,t},\mathcal F_{i,t}
\right),
\qquad
r_{i,t}\in[0,1].
\label{eq:belief-report}
\end{equation}
The platform observes $r_{i,t}$ but not $q_{i,t}$. Candidate
content is fixed before reporting, so the strategic action
studied here is the confidence report rather than candidate
generation.

Let $\mathcal A_t$ denote the feasible allocation set and
$x_{i,t}(a)\geq 0$ the public allocation value assigned to
agent $i$ under allocation $a$, with $x_{i,t}(a)=0$ when the
agent is not selected. For ranked allocation,
$x_{i,t}(a)$ may take a position-dependent value
$d_1\geq\cdots\geq d_K>0$. Given the report profile
$r_t=(r_{1,t},\ldots,r_{N,t})$, the platform selects
\begin{equation}
a_t(r_t)
\in
\arg\max_{a\in\mathcal A_t}
\sum_{j=1}^{N}
x_{j,t}(a)r_{j,t},
\label{eq:allocation-rule}
\end{equation}
using fixed, report-independent tie-breaking. The additive
objective in Eq.~\eqref{eq:allocation-rule} is the surrogate
optimized by the mechanism; task-specific downstream
metrics are evaluated separately.

In addition to outcome quality, agent $i$ may receive a
platform-controlled benefit $G_{i,t}(a_t(r_t))$ from the
resulting allocation. This benefit may represent current-round
exposure, routing credit, or decision influence. Because it
varies with the report-induced allocation, it can create an
incentive to exaggerate confidence. We assume that
$G_{i,t}$ is publicly specified, accurately modeled, and
expressed in the same utility units as the mechanism transfer.
Arbitrary external or unobserved benefits fall outside this
model.

Agents have quasi-linear, risk-neutral, current-round utility.
For a mechanism transfer $T_{i,t}$, agent $i$'s total utility is
\begin{equation}
U_{i,t}(r_t)
=
G_{i,t}\!\left(a_t(r_t)\right)
+
T_{i,t}(r_t,y_t,z_t),
\label{eq:quasilinear-utility}
\end{equation}
where the allocation benefit and transfer are expressed in the
same cardinal utility units.

\paragraph{Design objective.}
The platform seeks a transfer rule satisfying three properties.
First, for every fixed report profile $r_{-i,t}$ of the other
agents, reporting $r_{i,t}=q_{i,t}$ should be the unique
expected-utility best response. Second, truthful reports should
induce an allocation maximizing the belief-space surrogate:
\begin{equation}
a_t(q_t)
\in
\arg\max_{a\in\mathcal A_t}
\sum_{j=1}^{N}
x_{j,t}(a)q_{j,t}.
\label{eq:belief-space-allocation}
\end{equation}
Third, these expected incentives should remain valid when
outcomes are verified only through positive,
report-independent sparse audits. A fixed proper score alone
does not generally satisfy the first property because crossing
a hard allocation threshold can yield a discrete allocation
benefit.

\subsection{Audited Pivot Scoring}
\label{subsec:audited-pivot}

APSM realizes these objectives through an audited Groves
allocation structure together with strict probability
elicitation.

\paragraph{Pivot externality charge.}
Let $\mathcal A_t^{(-i)}$ be the feasible allocation set after
removing agent $i$. Define the maximum reported welfare that
the other agents can obtain without $i$ and the corresponding
Clarke pivot charge as
\begin{equation}
\begin{aligned}
H_{i,t}(r_{-i,t})
&=
\max_{a\in\mathcal A_t^{(-i)}}
\sum_{j\neq i}
x_{j,t}(a)r_{j,t},\\
p_{i,t}^{\mathrm{pivot}}(r_t)
&=
H_{i,t}(r_{-i,t})
-
\sum_{j\neq i}
x_{j,t}\!\left(a_t(r_t)\right)r_{j,t}.
\end{aligned}
\label{eq:pivot-charge}
\end{equation}
The pivot charge measures the loss in reported welfare imposed
on the remaining agents by agent $i$'s participation in the
allocation. Unlike a realized leave-one-out reward, the
counterfactual in Equation~\eqref{eq:pivot-charge}
re-optimizes the allocation after removing agent $i$.

\paragraph{Sparse outcome verification.}
The platform fixes an audit probability
$\epsilon_{i,t}>0$ using only pre-report public information.
After reports are submitted, it audits each candidate,
including unselected candidates, using fresh randomness:
\begin{equation}
z_{i,t}
\sim
\mathrm{Bernoulli}(\epsilon_{i,t}),
\qquad
z_{i,t}
\perp
\left(r_t,y_t,\mathcal F_{i,t}\right)
\mid
\epsilon_{i,t}.
\label{eq:audit-exogeneity}
\end{equation}
Report- and outcome-independent sampling gives
\begin{equation}
\begin{aligned}
\mathbb E\!\left[
\frac{z_{i,t}}{\epsilon_{i,t}}y_{i,t}
\,\middle|\,
\mathcal F_{i,t}
\right]
&=
q_{i,t},\\
\mathbb E\!\left[
\frac{z_{i,t}}{\epsilon_{i,t}}
S(r_{i,t},y_{i,t})
\,\middle|\,
\mathcal F_{i,t}
\right]
&=
\overline S(r_{i,t};q_{i,t}),
\end{aligned}
\label{eq:audit-unbiasedness}
\end{equation}
where
$\overline S(r;q)=qS(r,1)+(1-q)S(r,0)$.
Inverse-propensity weighting therefore preserves the expected
outcome and scoring credits, although decreasing
$\epsilon_{i,t}$ increases the variance and tail magnitude of
realized transfers.

\paragraph{APSM transfer.}
Let $S$ be a strictly proper scoring rule and let $\eta>0$
control its weight. APSM assigns the transfer
\begin{equation}
\begin{aligned}
T_{i,t}^{\mathrm{APSM}}
&=
-G_{i,t}\!\left(a_t(r_t)\right)
-
p_{i,t}^{\mathrm{pivot}}(r_t)\\
&\quad+
\frac{z_{i,t}}{\epsilon_{i,t}}
\left[
x_{i,t}\!\left(a_t(r_t)\right)y_{i,t}
+
\eta S(r_{i,t},y_{i,t})
\right].
\end{aligned}
\label{eq:APSM-transfer}
\end{equation}
In our implementation, we use the normalized Brier score
\begin{equation}
S(r,y)
=
1-(r-y)^2.
\label{eq:normalized-brier}
\end{equation}

Substituting Equation~\eqref{eq:APSM-transfer} into the
quasi-linear utility in Eq.~\eqref{eq:quasilinear-utility}
cancels the platform-controlled allocation benefit:
\begin{equation}
\begin{aligned}
U_{i,t}^{\mathrm{APSM}}
&=
-p_{i,t}^{\mathrm{pivot}}(r_t)\\
&\quad+
\frac{z_{i,t}}{\epsilon_{i,t}}
\left[
x_{i,t}\!\left(a_t(r_t)\right)y_{i,t}
+
\eta S(r_{i,t},y_{i,t})
\right].
\end{aligned}
\label{eq:APSM-cancelled-utility}
\end{equation}

Taking expectation over audit and outcome randomness under
the private belief $q_{i,t}$ gives
\begin{equation}
\begin{aligned}
\overline U_{i,t}^{\mathrm{APSM}}
&=
x_{i,t}\!\left(a_t(r_t)\right)q_{i,t}
-
p_{i,t}^{\mathrm{pivot}}(r_t)
+
\eta\overline S(r_{i,t};q_{i,t})\\
&=
x_{i,t}\!\left(a_t(r_t)\right)q_{i,t}
+
\sum_{j\neq i}
x_{j,t}\!\left(a_t(r_t)\right)r_{j,t}\\
&\quad-
H_{i,t}(r_{-i,t})
+
\eta\overline S(r_{i,t};q_{i,t}).
\end{aligned}
\label{eq:groves-form}
\end{equation}

Equation~\eqref{eq:groves-form} exposes the mechanism's
allocation-sensitive structure. Compensation removes the
direct benefit $G_{i,t}$, while the audited outcome credit and
pivot charge make the agent internalize the report-induced
allocation consequence. When agent $i$ reports $q_{i,t}$,
the induced allocation maximizes
\begin{equation}
x_{i,t}(a)q_{i,t}
+
\sum_{j\neq i}
x_{j,t}(a)r_{j,t}.
\label{eq:agent-belief-welfare}
\end{equation}
The strictly proper score then removes indifference among
reports that induce the same belief-efficient allocation and
uniquely identifies the full relevance probability.

Relative to Compensated Proper Scoring (CPS), which retains
benefit compensation and proper scoring but omits the audited
Groves allocation term, APSM adds a belief-space loss whenever
a deviation induces an inefficient allocation. It therefore
enlarges the utility separation against consequential,
allocation-changing reports rather than merely increasing the
penalty for probability error.

\paragraph{Agent-local utility evaluation.}
The platform need not observe the private belief $q_{i,t}$.
We evaluate two execution modes inside the agent-local trust
boundary.

In controlled replay, a mechanism-agnostic executor selects
\begin{equation}
\widehat r_{i,t}
\in
\arg\max_{r\in\mathcal R}
\mathbb E_{
Y\sim\mathrm{Bernoulli}(q_{i,t}),\,Z
}
\left[
U_{i,t}(r,r_{-i,t},Y,Z)
\right],
\label{eq:local-executor}
\end{equation}
where $\mathcal R$ is the admissible report domain. The same
executor is used for every mechanism; only the public utility
function changes. It releases $\widehat r_{i,t}$ without
exposing $q_{i,t}$ to the platform allocation procedure.

In LLM-controlled execution, the agent-local interface returns
the expected utilities of the admissible reports but neither
identifies the maximizer nor submits a report. The LLM retains
control of its final confidence. Both modes keep $q_{i,t}$
inside the agent-local trust boundary, but only controlled
replay imposes exact utility-maximizing execution. These
execution aids are implementation tools rather than part of
the incentive guarantee.

\section{Theoretical Analysis}
\label{sec:theory}

We analyze the reporting incentives induced by the
APSM.
We assume fixed candidate content, additive position-weighted
relevance welfare, platform-controlled selection benefits that
are exactly offset by the mechanism, report-independent audits,
and risk-neutral agents that maximize current-round expected
utility.
The scoring rule $S$ is $c_S$-strongly proper:
\begin{equation}
\bar S(q;q)-\bar S(r;q)
\ge
c_S(r-q)^2,
\qquad
\forall r,q\in[0,1].
\label{eq:strong-properness}
\end{equation}

\subsection{Why Matched-Brier Is Insufficient}

A proper score is truthful when it is the agent's only payoff.
This guarantee no longer holds when reports also determine a
hard selection decision.

\begin{proposition}[Threshold manipulation]
\label{prop:threshold-manipulation}
Suppose that agent $i$ is unselected when reporting its private
belief $q_i$, but becomes selected after reporting $q_i+d$,
where $d>0$ and $q_i+d\le 1$, obtaining an additional selection
benefit $b_i>0$.
Under the Matched-Brier utility
\begin{equation}
\bar U_i^{\mathrm{MB}}(r_i,r_{-i})
=
G_i\!\left(a(r_i,r_{-i})\right)
+
\eta\bar S(r_i;q_i),
\end{equation}
the deviation is profitable whenever
\begin{equation}
b_i>\eta d^2.
\label{eq:threshold-condition}
\end{equation}
Consequently, for every finite $\eta$ and positive selection
benefit, there exist report profiles arbitrarily close to the
Top-$K$ threshold for which truthful reporting is not optimal.
\end{proposition}

The result follows because the benefit from crossing the selection
threshold is discrete, whereas the Brier loss caused by a small
report deviation is quadratic.

\subsection{Truthful Reporting under APSM}

For fixed reports $r_{-i}$ of the other agents, define
\begin{equation}
\Phi_i(a;q_i,r_{-i})
=
x_i(a)q_i
+
\sum_{j\neq i}x_j(a)r_j.
\label{eq:true-welfare-objective}
\end{equation}
Using the transfer in Equation~\eqref{eq:APSM-transfer}, agent
$i$'s expected utility can be written as
\begin{equation}
\begin{aligned}
\bar U_i^{\mathrm{APSM}}
&=
\Phi_i\!\left(a(r);q_i,r_{-i}\right)
-
H_i(r_{-i}) \\
&\quad+
\eta\bar S(r_i;q_i),
\end{aligned}
\label{eq:APSM-groves-form}
\end{equation}
where $H_i(r_{-i})$ is independent of $r_i$.

Without the proper-score term, truthful reporting is a weak best
response because it induces an allocation maximizing
Equation~\eqref{eq:true-welfare-objective}; however, other reports
that induce the same allocation may be equally optimal.
The proper score removes this ambiguity.

\begin{theorem}[Strict expected truthfulness]
\label{thm:APSM-truthfulness}
For every private belief $q_i\in[0,1]$ and every report profile
$r_{-i}$ of the other agents, $r_i=q_i$ uniquely maximizes the
expected APSM utility.
Moreover, for any alternative report $r_i$,
\begin{equation}
\bar U_i^{\mathrm{APSM}}(q_i,r_{-i})
-
\bar U_i^{\mathrm{APSM}}(r_i,r_{-i})
\ge
\eta c_S(r_i-q_i)^2.
\label{eq:APSM-utility-gap}
\end{equation}
\end{theorem}

\paragraph{Proof sketch.}
Reporting $q_i$ makes the allocation rule maximize
$\Phi_i(a;q_i,r_{-i})$.
The pivot counterfactual $H_i(r_{-i})$ is independent of the
agent's report, while strong properness uniquely maximizes the
remaining score term at $r_i=q_i$.
Combining the allocation and scoring terms gives
Equation~\eqref{eq:APSM-utility-gap}.
The complete proof is provided in the appendix.

\subsection{Boundedly Rational Execution}

This section characterizes the utility-maximizing report, but an
LLM may only solve the induced optimization.
The agent-local executor makes this approximation explicit.

\begin{corollary}[Approximate execution]
\label{cor:executor-error}
Suppose the local executor returns $\hat r_i$ whose expected
utility is at most $\delta_i$ below the optimum:
\begin{equation}
\max_{r\in[0,1]}
\bar U_i^{\mathrm{APSM}}(r,r_{-i})
-
\bar U_i^{\mathrm{APSM}}(\hat r_i,r_{-i})
\le
\delta_i.
\end{equation}
Then
\begin{equation}
|\hat r_i-q_i|
\le
\sqrt{\frac{\delta_i}{\eta c_S}}.
\label{eq:executor-error}
\end{equation}
\end{corollary}

Finally, for any report- and outcome-measurable transfer component
$F_i$, inverse-propensity auditing preserves its expectation:
\begin{equation}
\mathbb E\!\left[
\frac{z_i}{\epsilon_i}F_i
\right]
=
\mathbb E[F_i].
\end{equation}
Hence the audit rate preserves the expected best response,
although the transfer variance grows as $O(1/\epsilon_i)$.
The guarantee requires accurately modeled platform-controlled
selection benefits; under deterministic Top-$K$, even a residual
benefit can reintroduce threshold manipulation near the cutoff.

\section{Experiments}
\label{sec:experiments}

We study three research questions.
\textbf{RQ1} examines whether allocation incentives distort
confidence reports from actual LLM agents and whether
compensation-based mechanisms can be executed through a
shared agent-local utility interface.
\textbf{RQ2} uses controlled replay to compare reporting and
allocation outcomes in two task-grounded workflows.
\textbf{RQ3} isolates what the audited Groves component adds
beyond Compensated Proper Scoring (CPS) under exact and
approximate execution.
Additional analyses of sparse audits, payoff misspecification,
finite report grids, and approximate optimization appear in
the appendix.

\subsection{Experimental Setup}
\label{sec:exp_setup}

\noindent\textbf{Evaluation environments.}
We evaluate two six-agent environments using the same mixed
population of two Qwen3-8B, two Qwen3-4B
\cite{yang2025qwen3}, and two DeepSeek-V4-Pro agents
\cite{xu2026deepseek}, with three paired seeds. AgentWebBench
\cite{zhong2026agentwebbench} contains 200 web-retrieval
queries with Top-$2$ content-agent allocation and document-level
NDCG@5 evaluation. Our MultiAgentBench database extension
\cite{zhu2025multiagentbench} contains 50 PostgreSQL incidents
covering ten single- and dual-root-cause configurations, where
specialist agents communicate for two rounds before Top-$2$
escalation.

\noindent\textbf{Reference beliefs and execution.}
For agent $i$, $q_i$ is the frozen agent-local belief formed
before mechanism-dependent reporting and is not observed by
the platform. Controlled replay reuses the same candidates,
evidence, outcomes, report domains, allocation rules,
tie-breaking, audit randomness, and optimizer across
mechanisms. Exact execution denotes exact expected-utility
maximization over the finite admissible report domain.

We separately evaluate LLM-controlled utility-assisted
execution. A shared agent-local interface returns expected
utilities for admissible reports but neither identifies the
maximizer nor submits a report; the LLM retains control of
its final confidence. Matched-Brier, CPS, and APSM use
identical instances, prompts, tool access, and parsing rules,
differing only in their utility functions.

\noindent\textbf{Baselines.}
Selection-only includes the modeled allocation benefit.
Matched-Brier adds a scale-matched proper score without
compensation~\cite{gneiting2007strictly}. Difference Reward
uses a realized leave-one-out contribution
\cite{wolpert2001optimal}, while Pivot-only retains
compensation, audited outcome value, and the Clarke charge
but omits proper scoring
\cite{clarke1971multipart,groves1973incentives}.
CPS follows \citet{boutilier2011eliciting}:
\begin{equation}
T_i^{\mathrm{CPS}}
=
-G_i(a(r))
+
\eta\frac{z_i}{\epsilon_i}S(r_i,y_i).
\label{eq:cps-transfer-exp}
\end{equation}
It uses the same compensation and scoring parameters as
APSM. Allocation induced directly by $q$ serves only as a
non-deployable reference.

\noindent\textbf{Metrics.}
Reporting fidelity is measured by
\begin{equation}
\mathrm{RefDev}
=
\frac{1}{|\mathcal D|N}
\sum_{t\in\mathcal D}
\sum_{i=1}^{N}
|r_{i,t}-q_{i,t}|.
\label{eq:refdev}
\end{equation}
We also report all-candidate Brier loss, cutoff-crossing rate,
and exact reference-allocation match. For utility-assisted
execution, the utility-consistent choice rate is
\begin{equation}
\mathrm{UCR}
=
\frac{1}{|\mathcal D|N}
\sum_{t\in\mathcal D}
\sum_{i=1}^{N}
\mathbb I\!\left[
\max_{r\in\mathcal R_{i,t}}
\overline U_{i,t}(r)
-
\overline U_{i,t}(\widehat r_{i,t})
<10^{-6}
\right].
\label{eq:utility-consistent-choice}
\end{equation}
Downstream quality is measured by NDCG@5 for
AgentWebBench and Top-$1$, Top-$2$, and MRR for database
diagnosis. RQ3 additionally reports the utility margin against
allocation-changing reports.

\noindent\textbf{Statistical analysis.}
All comparisons are paired. LLM experiments use query-level
bootstrap; controlled AgentWebBench replay uses 10,000
seed--query hierarchical bootstrap draws. Database results use
5,000 incident-cluster bootstrap draws after averaging paired
differences over three specialist--model rotations. Agents
within the same query or incident are not treated as
independent observations.

\subsection{RQ1: Do Allocation Incentives Create a
Strategic-Calibration Failure?}
\label{sec:rq1}

\noindent\textbf{Calculator-free strategic reporting.}
We freeze each candidate, supporting evidence, and initial confidence,
then ask the same LLM to report under different public payoff conditions.
No local optimizer or utility tool is used. A neutral re-report controls
for ordinary elicitation variation.

Table~\ref{tab:calculator_free_reporting} shows that
Selection-only raises confidence by $0.218$ relative to Neutral
(95\% CI $[0.196,0.249]$), and $95.8\%$ of eligible agents cross the
Top-$K$ cutoff. Matched-Brier reduces the shift to $0.070$, but
$51.6\%$ still cross. Thus, proper scoring attenuates but does not remove
allocation-induced inflation under deterministic Top-$K$ selection.

\noindent\textbf{Tool-assisted LLM execution.}
We next test whether actual LLM agents can execute the
mechanisms when utility calculations are available while the
final report remains under LLM control.
Table~\ref{tab:tool_assisted_reporting} shows that both CPS
and APSM sharply reduce distortion relative to tool-assisted
Matched-Brier. RefDev decreases from $0.325$ to $0.007$
and $0.014$, respectively, while the cutoff-crossing rate
falls from $75.3\%$ to $2.1\%$ and $0.8\%$.

The two compensated mechanisms exhibit complementary
strengths. CPS achieves lower RefDev, a higher exact-match
rate ($92.5\%$ versus $85.5\%$; APSM--CPS:
$-7.0$ percentage points, 95\% CI $[-12.5,-1.5]$), and a
higher utility-consistent choice rate. APSM obtains lower
Brier loss ($0.403$ versus $0.412$; paired difference
$-0.0094$, 95\% CI $[-0.0132,-0.0059]$) and numerically
fewer cutoff crossings. Overall,
both compensation-based mechanisms can be executed by
actual LLM agents, while their relative advantages depend on
the reporting and allocation metric considered.

\begin{table}[t]
\centering
\small
\setlength{\tabcolsep}{2.2pt}
\begin{tabular*}{\columnwidth}{
@{\extracolsep{\fill}}lcccc@{}}
\toprule
\textbf{Condition}
& \shortstack{\textbf{Shift}\\\textbf{(95\% CI)}}
& \shortstack{\textbf{RefDev}\\$\downarrow$}
& \shortstack{\textbf{Cross.}\\$\downarrow$}
& \shortstack{\textbf{Parse err.}\\$\downarrow$} \\
\midrule
Neutral
& $0.000$
& $0.185$
& $0.0\%$
& $0.33\%$ \\

Selection
& \shortstack{$+0.218$\\$[+0.196,+0.249]$}
& $0.387$
& $95.8\%$
& $0.00\%$ \\

M-Brier
& \shortstack{$+0.070$\\$[+0.053,+0.088]$}
& $0.225$
& $51.6\%$
& $0.67\%$ \\
\bottomrule
\end{tabular*}
\caption{Calculator-free LLM reporting on AgentWebBench
(200 queries; 1,200 reports per condition). Shift is relative to Neutral.}
\label{tab:calculator_free_reporting}
\end{table}

\begin{table}[t]
\centering
\small
\setlength{\tabcolsep}{1.8pt}
\begin{tabular*}{\columnwidth}{
@{\extracolsep{\fill}}lccccc@{}}
\toprule
\textbf{Method}
& \shortstack{\textbf{RefDev}\\$\downarrow$}
& \shortstack{\textbf{Cross.}\\$\downarrow$}
& \shortstack{\textbf{Brier}\\$\downarrow$}
& \shortstack{\textbf{Match}\\$\uparrow$}
& \shortstack{\textbf{U-cons.}\\\textbf{choice $\uparrow$}} \\
\midrule
M-Brier
& $0.325$
& $75.3\%$
& $0.729$
& $0.0\%$
& $78.3\%$ \\

CPS
& $\mathbf{0.007}$
& $2.1\%$
& $0.412$
& $\mathbf{92.5\%}$
& $\mathbf{93.7\%}$ \\

APSM
& $0.014$
& $\mathbf{0.8\%}$
& $\mathbf{0.403}$
& $85.5\%$
& $86.7\%$ \\
\bottomrule
\end{tabular*}
\caption{Tool-assisted LLM execution on 200 paired
AgentWebBench queries. The shared interface returns report
utilities but neither identifies nor submits the maximizing
report. All methods achieve a 100\% tool-invocation rate;
U-cons.\ choice denotes the utility-consistent choice rate.}
\label{tab:tool_assisted_reporting}
\end{table}

\noindent\textbf{Interaction with collaborative calibration.}
Table~\ref{tab:cc_main} separates belief formation from
mechanism-dependent reporting. Collaborative Calibration
improves the frozen reference beliefs, reducing Brier loss from
$0.414$ to $0.198$, but lowers their induced NDCG@5 from
$0.453$ to $0.364$. Holding the revised beliefs fixed,
Matched-Brier introduces RefDev $0.221$, further reduces
NDCG@5 to $0.301$, and matches the CC-induced allocation
on only $2.5\%$ of queries. In controlled replay, APSM
recovers the revised reports and their induced allocation,
achieving zero RefDev, NDCG@5 of $0.364$, and a $100\%$
exact-match rate. Thus, collaborative calibration improves
belief quality, while incentive-compatible reporting prevents
subsequent allocation-induced distortion.

\begin{table}[t]
\centering
\caption{CC followed by strategic reporting. Results use
200 queries, three paired seeds, and six agents per query.}
\label{tab:cc_main}
\small
\setlength{\tabcolsep}{3.2pt}
\renewcommand{\arraystretch}{1.08}
\begin{tabular}{llcccc}
\toprule
\textbf{Stage}
& \textbf{Method}
& \textbf{Brier $\downarrow$}
& \textbf{RefDev $\downarrow$}
& \textbf{NDCG@5 $\uparrow$}
& \textbf{Match $\uparrow$} \\
\midrule
Belief
& Orig. ref.
& 0.414 & -- & 0.453 & -- \\

Belief
& CC ref.
& 0.198 & -- & 0.364 & -- \\
\midrule
Report
& M-Brier
& -- & 0.221 & 0.301 & 0.025 \\

Report
& APSM
& -- & \textbf{0.000}
& \textbf{0.364}
& \textbf{1.000} \\
\bottomrule
\end{tabular}
\end{table}

\begin{figure*}[ht]
    \centering
    \includegraphics[width=0.9\textwidth]
    {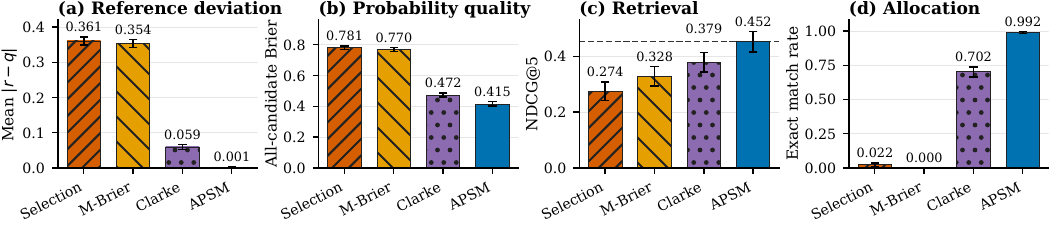}
    \caption{Controlled AgentWebBench replay over 200 queries and three
    seeds: \textbf{(a)} RefDev, \textbf{(b)} all-candidate Brier,
    \textbf{(c)} NDCG@5 (reference-belief oracle: 0.453), and \textbf{(d)} exact reference-allocation
    match. Error bars are 95\% hierarchical-bootstrap intervals.}
    \label{fig:agentweb_main}
\end{figure*}

\subsection{RQ2: How Do the Mechanisms Behave under
Controlled Task-Grounded Replay?}
\label{sec:rq2}

Having evaluated LLM-controlled execution in RQ1, we use controlled
replay to verify implementation consistency and decompose mechanism
behavior. The same mechanism-agnostic optimizer is used throughout;
these results are controlled mechanism analyses rather than evidence of
unaided LLM behavior.

\noindent\textbf{Content-agent allocation.}
Figure~\ref{fig:agentweb_main} summarizes the controlled
AgentWebBench replay over 200 queries and three paired seeds.
Matched-Brier retains RefDev $0.354$ and never reproduces the
reference-induced allocation, whereas APSM reduces RefDev to
$0.001$ and reaches $99.2\%$ exact match. APSM also improves
NDCG@5 from $0.328$ to $0.452$, within $0.001$ of the
reference-belief oracle. Pivot-only improves allocation
alignment but retains non-unique probability reports.

\noindent\textbf{Database root-cause diagnosis.}
In the MultiAgentBench workflow, six specialists compete for two
escalation positions. Table~\ref{tab:database_main} shows that controlled
APSM replay reduces RefDev from $0.339$ under Matched-Brier to zero,
improves Top-$2$ success from $0.627$ to $0.907$, and recovers the
reference-induced allocation. For dual-root incidents, Top-$1$ success
increases from $0.480$ to $1.000$. Together, the two workflows show how
the mechanism components affect ranked allocation under heterogeneous
agents, tool use, communication, and multiple root causes.

\begin{table*}[t]
\centering
\caption{Controlled MultiAgentBench database diagnosis over three paired
specialist--model rotations. All mechanisms reuse the same diagnoses,
evidence, beliefs, report set, Top-$2$ router, and optimizer.}
\label{tab:database_main}
\small
\setlength{\tabcolsep}{5.2pt}
\begin{tabular}{lcccccc}
\toprule
Mechanism &
Ref. Dev. $\downarrow$ &
Brier $\downarrow$ &
Top-2 $\uparrow$ &
Top-1 $\uparrow$ &
MRR $\uparrow$ &
Exact Ref. Match $\uparrow$ \\
\midrule
Selection-only
& 0.687 & 0.735 & 0.500 & 0.300 & 0.545 & 0.220 \\
Matched-Brier
& 0.339 & 0.378 & 0.627 & 0.567 & 0.678 & 0.000 \\
Difference Reward
& 0.344 & 0.382 & 0.827 & 0.740 & 0.823 & 0.193 \\
Pivot-only
& 0.531 & 0.546 & 0.880 & \textbf{0.853}
& 0.897 & 0.747 \\
APSM
& \textbf{0.000} & \textbf{0.057} & \textbf{0.907}
& \textbf{0.853} & \textbf{0.903} & \textbf{1.000} \\
\bottomrule
\end{tabular}
\end{table*}

\begin{figure*}[ht]
    \centering
    \includegraphics[width=1\textwidth]
    {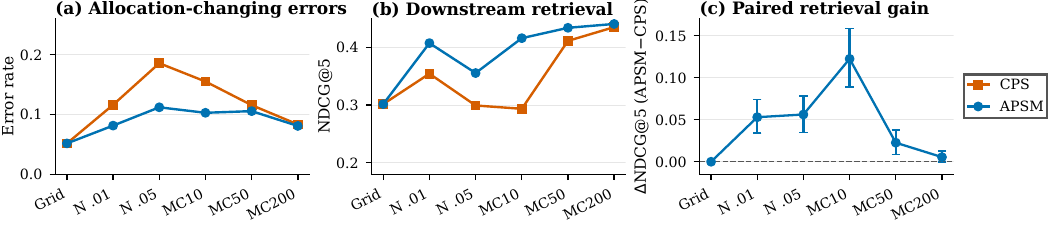}
    \caption{CPS and APSM under approximate AgentWebBench execution:
    \textbf{(a)} allocation-changing error rate,
    \textbf{(b)} NDCG@5, and
    \textbf{(c)} paired APSM--CPS NDCG@5 gain with 95\% confidence
    intervals.}
    \label{fig:cps_comparison}
\end{figure*}

\subsection{RQ3: What Does the Audited Groves Component
Add Beyond Compensated Proper Scoring?}
\label{sec:rq3}

CPS and APSM share exact influence compensation and the same
strictly proper score; their comparison isolates the audited
realized-value and pivot-externality terms.

\noindent\textbf{Exact execution.}
For this diagnostic, each report domain is augmented to include
the frozen belief $q_i$, and expected utility is maximized exactly
over the resulting finite set. CPS and APSM both achieve zero
RefDev, $100\%$ exact reference-allocation match, and
NDCG@5 of $0.453$, consistent with the classical
compensation-rule guarantee~\cite{boutilier2011eliciting}.
Thus, APSM's additional value lies not in exact truthfulness,
but in its allocation-sensitive utility margin under approximate
execution.

\noindent\textbf{Allocation-sensitive margin.}
We define
\begin{equation}
M_i^{\mathrm{alloc}}
=
\overline U_i(q_i,r_{-i})
-
\max_{\substack{r_i\in\mathcal R:\\
a(r_i,r_{-i})\neq a(q_i,r_{-i})}}
\overline U_i(r_i,r_{-i}).
\end{equation}
APSM increases the mean margin from $0.3350$ under CPS to
$0.6462$ ($1.93\times$). Beyond the proper-score loss in CPS,
an allocation-changing deviation under APSM also incurs the
belief-space welfare loss of the induced allocation.

\noindent\textbf{Approximate execution.}
Figure~\ref{fig:cps_comparison} compares the mechanisms under
grid discretization, additive utility noise, and Monte Carlo
estimation. At noise level $0.05$, APSM reduces
allocation-changing errors from $18.6\%$ to $11.2\%$ and
increases NDCG@5 from $0.299$ to $0.356$. Its advantage
diminishes as utility estimates improve, consistent with the
exact-execution results. This also shows why small report errors
can remain consequential.

\noindent\textbf{Controlling for utility scale.}
Figure~\ref{fig:utility_scale_control} shows that APSM's advantage
is not due solely to a larger utility range. After matching the mean
utility range of CPS to APSM with $\eta=1.43$, APSM still reduces
allocation errors from 17.2\% to 11.2\%, cutoff crossings from
14.3\% to 8.3\%, and improves NDCG@5 from 0.304 to 0.356.
Under 5\% utility-range-normalized noise, APSM further reduces
allocation errors from 17.1\% to 10.5\% and improves NDCG@5
from 0.284 to 0.341.

\begin{figure*}[t]
    \centering
    \includegraphics[width=1\textwidth]{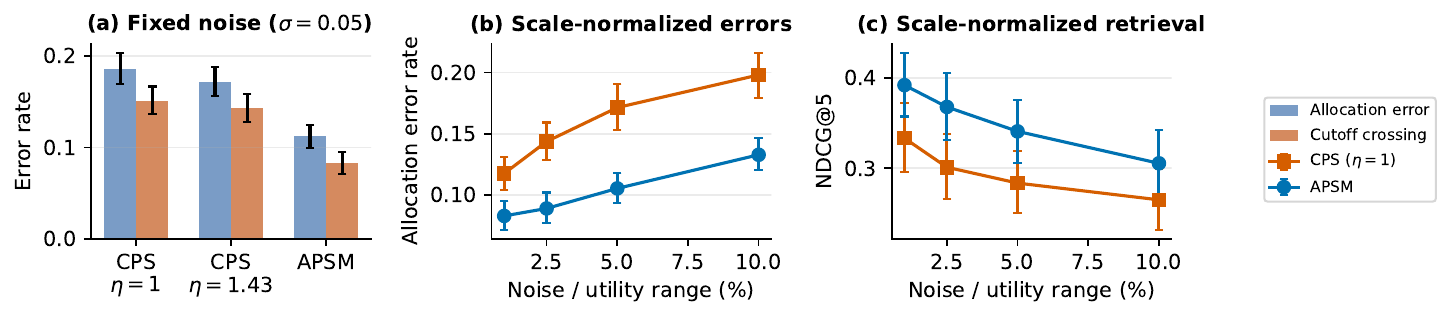}
    \caption{Utility-scale-controlled comparison of CPS and APSM.
    \textbf{(a)} Results under fixed utility noise.
    \textbf{(b)--(c)} Allocation error and NDCG@5 under
    utility-range-normalized noise. Error bars denote 95\% seed--query hierarchical-bootstrap
confidence intervals. }
    \label{fig:utility_scale_control}
\end{figure*}

\noindent\textbf{Payoff misspecification.}
At a representative $\pm5\%$ error in the modeled allocation
benefit, APSM retains a $46.5\%$ exact-match rate, compared
with approximately $15\%$ for CPS. APSM therefore degrades
more gradually, although neither mechanism is robust to
arbitrary unmodeled benefits.

\noindent\textbf{Sparse-audit trade-off.}
Sparse auditing preserves expected incentives but increases
settlement risk. At $\epsilon=0.1$, transfer variance reaches
$15.4\times$ and $30.1\times$ the full-audit level on
AgentWebBench and database diagnosis, with zero-audit rates
of $53.7\%$ and $52.9\%$. The complete sweeps appear in the
appendix.

\section{Related Work}

\paragraph{LLM calibration and strategic multi-agent behavior.}
Prior work improves confidence quality through post-hoc
calibration, prompted confidence elicitation, and
multi-calibration~\cite{guo2017calibration,jiang2021can,
kadavath2022language,lin2022teaching,
hebert2018multicalibration,detommaso2024multicalibration}.
Multi-agent approaches further refine confidence through
deliberation and critique~\cite{pandey2026refine,
yang2024confidence}, while broader LLM multi-agent systems
use collaboration, debate, and task decomposition to improve
reasoning~\cite{du2023improving,wu2023autogen,
park2023generative}. These methods generally assume
cooperative or aligned agents and treat confidence as a passive
expression of uncertainty. Recent evidence, however, shows
that LLM agents may adapt their behavior to persuasive or
utility-seeking objectives~\cite{agarwal2025persuasion,
su2025ai}. We study a complementary reporting-layer
problem: even with meaningful
agent-local beliefs, allocation incentives can distort the
confidence submitted to the coordinator.

\paragraph{Truthful elicitation and decision-dependent allocation.}
Truthful probabilistic reporting has been studied through
proper scoring rules~\cite{gneiting2007strictly}, Groves
mechanisms~\cite{vickrey1961counterspeculation,
clarke1971multipart,groves1973incentives}, peer
prediction~\cite{prelec2004bayesian,miller2005eliciting,
shnayder2016informed}, and multi-agent forecast
aggregation~\cite{wang2021forecast,lambert2008eliciting}.
Closest to our model, \citet{boutilier2011eliciting} shows
that compensating experts for forecast-induced decisions,
together with a strictly proper score, achieves strict
elicitation under exact compensation. We instantiate this
mechanism as CPS and do not claim novelty for exact
truthfulness. Decision-market research further studies
forecasts that affect both the chosen action and the observed
outcome~\cite{othman2010decision,chen2011information,
chen2014eliciting,oesterheld2020decision}. Our setting
instead concerns LLM agents whose relevance reports
determine scarce deterministic Top-$K$ allocation. APSM
preserves this allocation rule while introducing an audited
Groves term that enlarges the utility margin against
allocation-changing deviations under approximate execution.
Relatedly, \citet{doval2025calibrated} study mechanism
design by principals with calibrated beliefs about agent
types, whereas we elicit reports from strategic agents under
sparse direct verification.

\section{Conclusion and Limitations}
\label{sec:conclusion}

We study strategic confidence reporting in LLM-based
multi-agent systems, where confidence affects both uncertainty
expression and downstream allocation. We introduce APSM,
an audited allocation-sensitive elicitation mechanism that
offsets platform-controlled benefits, internalizes
report-induced externalities, and uniquely elicits relevance
probabilities. Calculator-free experiments show systematic
confidence inflation under hard allocation incentives. With
agent-local utility assistance, actual LLM agents substantially
reduce this distortion under both CPS and APSM, while
controlled replay shows that APSM enlarges the utility margin
against allocation-changing deviations and improves robustness
under approximate execution. Our guarantees assume fixed candidate content, risk-neutral
myopic agents, report-independent auditing, and accurately
modeled platform-controlled benefits. The reference belief is
an agent-local proxy rather than observable truth, and APSM is
not uniformly easier to execute than CPS. The framework does not yet address external benefits, collusion,
budget balance, or strategic generation; broader strategic settings remain future work.

\bibliography{neurips_2026}

\clearpage
\appendix
\section{Appendix}
\label{app:full-mechanism}

This section provides the complete formalization of the
\textbf{Audited Pivot Scoring Mechanism (APSM)}.
We first specify the repeated coordination protocol and information
structure, and then derive the complete transfer and expected-utility
representation used by the theoretical analysis.

\subsection{Repeated Coordination Protocol}
\label{app:protocol}

We consider a repeated interaction between a platform and a
persistent population of $N$ content agents over rounds
$t=1,\dots,T$.
At round $t$, the environment provides a query
$x_t\in\mathcal X$.
Each agent $i\in[N]$ generates a candidate item
$c_{i,t}$, such as a retrieved document, evidence passage, or
candidate answer.

Candidate generation is completed before confidence elicitation.
The candidate set
\begin{equation}
\mathcal C_t
=
\{c_{1,t},\dots,c_{N,t}\}
\label{eq:app-candidate-set}
\end{equation}
is therefore fixed when reports are submitted.
The strategic action studied in this work is the confidence report,
rather than the candidate-generation policy.

Each candidate has a binary relevance outcome
\begin{equation}
y_{i,t}\in\{0,1\},
\label{eq:app-relevance-outcome}
\end{equation}
where $y_{i,t}=1$ indicates that candidate $c_{i,t}$ is relevant to
the current query.
Agent $i$ forms a private relevance belief
\begin{equation}
q_{i,t}
:=
\Pr\!\left(
y_{i,t}=1
\mid
x_t,c_{i,t},\mathcal F_{i,t}
\right),
\label{eq:app-private-belief}
\end{equation}
where $\mathcal F_{i,t}$ denotes its local information state.
The platform does not observe $q_{i,t}$.
Instead, the agent submits a report
\begin{equation}
r_{i,t}\in[0,1].
\label{eq:app-report}
\end{equation}

We write
\begin{equation}
q_t=(q_{1,t},\dots,q_{N,t}),
\qquad
r_t=(r_{1,t},\dots,r_{N,t}).
\label{eq:app-belief-report-vectors}
\end{equation}

\subsection{Allocation and Surrogate Welfare}
\label{app:allocation}

Let $\mathcal A_t$ denote the set of feasible allocations at round
$t$.
An allocation $a\in\mathcal A_t$ assigns at most $K$ agents to
ranked coordination positions.
For agent $i$, define
\begin{equation}
x_{i,t}(a)
\in
\{0,d_1,\dots,d_K\},
\label{eq:app-position-value}
\end{equation}
where
\begin{equation}
d_1\ge d_2\ge\dots\ge d_K>0
\end{equation}
is a position-dependent value and $x_{i,t}(a)=0$ if agent $i$ is
not selected.

Given the report profile $r_t$, the platform evaluates an allocation
through the additive report-space welfare
\begin{equation}
\widetilde W_t(a;r_t)
=
\sum_{j=1}^{N}
x_{j,t}(a)r_{j,t}.
\label{eq:app-reported-welfare}
\end{equation}
The selected allocation is
\begin{equation}
a_t(r_t)
\in
\arg\max_{a\in\mathcal A_t}
\widetilde W_t(a;r_t).
\label{eq:app-allocation-rule}
\end{equation}
Ties are resolved according to a fixed rule that does not depend on
the identity or report history of the tied agents.

The corresponding belief-space surrogate welfare is
\begin{equation}
\widetilde W_t(a;q_t)
=
\sum_{j=1}^{N}
x_{j,t}(a)q_{j,t}.
\label{eq:app-belief-welfare}
\end{equation}
The mechanism seeks to induce reports that make
Equation~\eqref{eq:app-allocation-rule} maximize
Equation~\eqref{eq:app-belief-welfare}.

\subsection{Platform-Controlled Selection Benefit}
\label{app:selection-benefit}

An agent may obtain a non-outcome benefit through selection,
exposure, traffic, or future platform-controlled opportunities.
We denote this benefit by
\begin{equation}
G_{i,t}(a).
\label{eq:app-selection-benefit}
\end{equation}

The main theoretical guarantee assumes that $G_{i,t}$ is publicly
specified and controlled by the platform.
For example, it may represent an exposure credit, a query quota, or
a platform-internal selection reward.
The mechanism explicitly offsets this benefit in its transfer.

This assumption does not cover external benefits that the
platform cannot observe or control. Sensitivity to a
misspecified selection benefit is analyzed below.

\subsection{Pivot Externality Charge}
\label{app:pivot-charge}

For agent $i$, let $\mathcal A_t^{(-i)}$ denote the feasible
allocation set after removing agent $i$.
The maximum report-space welfare attainable by the remaining agents
is
\begin{equation}
H_{i,t}(r_{-i,t})
=
\max_{a\in\mathcal A_t^{(-i)}}
\sum_{j\neq i}
x_{j,t}(a)r_{j,t}.
\label{eq:app-counterfactual-welfare}
\end{equation}

The Clarke pivot charge is
\begin{equation}
\begin{aligned}
p_{i,t}^{\mathrm{pivot}}(r_t)
&=
H_{i,t}(r_{-i,t}) \\
&\quad-
\sum_{j\neq i}
x_{j,t}\!\left(a_t(r_t)\right)r_{j,t}.
\end{aligned}
\label{eq:app-pivot-charge}
\end{equation}
It measures the report-space opportunity cost imposed on the other
agents by the participation and allocation of agent $i$.

The counterfactual in
Equation~\eqref{eq:app-counterfactual-welfare} re-optimizes the
allocation after removing $i$.
It is therefore different from a realized leave-one-out reward that
merely deletes an item from an already chosen set.

\subsection{Sparse Outcome Verification}
\label{app:sparse-verification}

The platform verifies the relevance outcome of candidate $i$ with
probability $\epsilon_{i,t}>0$.
Let
\begin{equation}
z_{i,t}
\sim
\mathrm{Bernoulli}(\epsilon_{i,t})
\label{eq:app-audit-indicator}
\end{equation}
denote the audit indicator.

Audit sampling is performed after report submission and is
conditionally independent of the current reports and outcomes:
\begin{equation}
z_{i,t}
\perp
(r_t,y_t)
\mid
\mathcal F_{t-1},x_t.
\label{eq:app-audit-exogeneity}
\end{equation}

The inverse-propensity realized-value credit is
\begin{equation}
V_{i,t}^{\mathrm{out}}
=
\frac{z_{i,t}}{\epsilon_{i,t}}
x_{i,t}\!\left(a_t(r_t)\right)y_{i,t}.
\label{eq:app-outcome-credit}
\end{equation}
Under the private belief $q_{i,t}$,
\begin{equation}
\mathbb E
\left[
V_{i,t}^{\mathrm{out}}
\mid
\mathcal F_{i,t}
\right]
=
x_{i,t}\!\left(a_t(r_t)\right)q_{i,t}.
\label{eq:app-outcome-credit-expectation}
\end{equation}

\subsection{Strictly Proper Score}
\label{app:proper-score}

Let
\begin{equation}
S:[0,1]\times\{0,1\}\rightarrow\mathbb R
\end{equation}
be a strictly proper scoring rule.
For a private belief $q$, define the expected score
\begin{equation}
\overline S(r;q)
=
qS(r,1)+(1-q)S(r,0).
\label{eq:app-expected-score}
\end{equation}

We assume that $S$ is $c_S$-strongly proper:
\begin{equation}
\overline S(q;q)
-
\overline S(r;q)
\ge
c_S(r-q)^2,
\qquad
\forall r,q\in[0,1].
\label{eq:app-strong-properness}
\end{equation}

The sparse proper-score credit is
\begin{equation}
V_{i,t}^{\mathrm{score}}
=
\eta
\frac{z_{i,t}}{\epsilon_{i,t}}
S(r_{i,t},y_{i,t}),
\label{eq:app-score-credit}
\end{equation}
where $\eta>0$ controls the strength of probability elicitation.

In our implementation, we use the normalized Brier score
\begin{equation}
S(r,y)
=
1-(r-y)^2.
\label{eq:app-brier-score}
\end{equation}
For this normalization,
\begin{equation}
\overline S(q;q)-\overline S(r;q)
=
(r-q)^2,
\end{equation}
and therefore $c_S=1$.

\subsection{Complete APSM Transfer}
\label{app:complete-transfer}

The APSM transfer is
\begin{equation}
\begin{aligned}
T_{i,t}^{\mathrm{APSM}}
&=
-
G_{i,t}\!\left(a_t(r_t)\right)
+
V_{i,t}^{\mathrm{out}} \\
&\quad-
p_{i,t}^{\mathrm{pivot}}(r_t)
+
V_{i,t}^{\mathrm{score}}.
\end{aligned}
\label{eq:app-APSM-transfer}
\end{equation}

Agent $i$ also receives the original platform-controlled selection
benefit $G_{i,t}(a_t(r_t))$.
Its total utility is therefore
\begin{equation}
U_{i,t}^{\mathrm{APSM}}
=
G_{i,t}\!\left(a_t(r_t)\right)
+
T_{i,t}^{\mathrm{APSM}}.
\label{eq:app-total-utility}
\end{equation}

Substituting Equation~\eqref{eq:app-APSM-transfer} into
Equation~\eqref{eq:app-total-utility} gives
\begin{equation}
\begin{aligned}
U_{i,t}^{\mathrm{APSM}}
&=
\frac{z_{i,t}}{\epsilon_{i,t}}
x_{i,t}\!\left(a_t(r_t)\right)y_{i,t} \\
&\quad-
p_{i,t}^{\mathrm{pivot}}(r_t)
+
\eta
\frac{z_{i,t}}{\epsilon_{i,t}}
S(r_{i,t},y_{i,t}).
\end{aligned}
\label{eq:app-cancelled-utility}
\end{equation}

Taking expectation with respect to both audit and outcome randomness
under the private belief $q_{i,t}$ yields
\begin{equation}
\begin{aligned}
\overline U_{i,t}^{\mathrm{APSM}}
&=
x_{i,t}\!\left(a_t(r_t)\right)q_{i,t}
-
p_{i,t}^{\mathrm{pivot}}(r_t) \\
&\quad+
\eta\overline S(r_{i,t};q_{i,t}).
\end{aligned}
\label{eq:app-expected-utility}
\end{equation}

Using Equation~\eqref{eq:app-pivot-charge}, we obtain
\begin{equation}
\begin{aligned}
\overline U_{i,t}^{\mathrm{APSM}}
&=
x_{i,t}\!\left(a_t(r_t)\right)q_{i,t}
+
\sum_{j\neq i}
x_{j,t}\!\left(a_t(r_t)\right)r_{j,t} \\
&\quad-
H_{i,t}(r_{-i,t})
+
\eta\overline S(r_{i,t};q_{i,t}).
\end{aligned}
\label{eq:app-groves-form}
\end{equation}

The term $H_{i,t}(r_{-i,t})$ is independent of $r_{i,t}$.
Equation~\eqref{eq:app-groves-form} is the Groves-form
representation used throughout the theoretical analysis.

\subsection{Agent-Local Policy Optimizer}
\label{app:executor}

The agent-local executor receives:

\begin{enumerate}[leftmargin=*,itemsep=2pt]
\item the agent's private belief $q_{i,t}$;
\item the public reports $r_{-i,t}$ of the other agents;
\item the public allocation and transfer rules;
\item the audit probability $\epsilon_{i,t}$; and
\item a finite candidate report set $\mathcal R$.
\end{enumerate}

It evaluates the expected utility of each candidate report:
\begin{equation}
\overline U_{i,t}(r)
=
\mathbb E_{
Y\sim\mathrm{Bernoulli}(q_{i,t}),\,Z
}
\left[
U_{i,t}(r,r_{-i,t},Y,Z)
\right].
\label{eq:app-report-utility}
\end{equation}
The submitted report is
\begin{equation}
\widehat r_{i,t}
\in
\arg\max_{r\in\mathcal R}
\mathcal U_{i,t}(r).
\label{eq:app-local-argmax}
\end{equation}

The same executor is used for all mechanisms.
Only the public utility function changes.
The executor therefore does not encode a rule such as
$\widehat r_{i,t}=q_{i,t}$.

The executor runs within the agent's local trust boundary.
Only $\widehat r_{i,t}$ is released to the platform.
The private belief $q_{i,t}$ is neither exposed to nor used by the
platform allocation procedure.

\section{Proofs of the Main Results}
\label{app:proofs}

\subsection{Proof of Proposition~\ref{prop:threshold-manipulation}}
\label{app:proof-threshold}

\begin{proof}
Under the normalized Brier score,
\begin{equation}
\overline S(q_i;q_i)
-
\overline S(q_i+d;q_i)
=
d^2.
\label{eq:app-brier-deviation-cost}
\end{equation}

By assumption, reporting $q_i+d$ changes the allocation so that
agent $i$ receives an additional selection benefit $b_i$.
The expected utility difference between the deviating and truthful
reports is therefore
\begin{equation}
\begin{aligned}
&
\overline U_i^{\mathrm{MB}}(q_i+d,r_{-i})
-
\overline U_i^{\mathrm{MB}}(q_i,r_{-i}) \\
&\qquad=
b_i-\eta d^2.
\end{aligned}
\label{eq:app-threshold-gain}
\end{equation}
The deviation is profitable whenever
\begin{equation}
b_i>\eta d^2.
\end{equation}

For any finite $\eta$ and any $b_i>0$, one can select
\begin{equation}
0<d<\sqrt{\frac{b_i}{\eta}}
\end{equation}
and construct a report profile in which agent $i$ lies within $d$
of the Top-$K$ cutoff.
Hence there always exist sufficiently close threshold profiles for
which truthful reporting is not optimal.
\end{proof}

\subsection{Pivot-Only Allocation Alignment}
\label{app:pivot-only}

For fixed $r_{-i,t}$, define
\begin{equation}
\Phi_{i,t}(a;q_{i,t},r_{-i,t})
=
x_{i,t}(a)q_{i,t}
+
\sum_{j\neq i}
x_{j,t}(a)r_{j,t}.
\label{eq:app-true-allocation-objective}
\end{equation}

\begin{proposition}[Pivot-only allocation alignment]
\label{prop:app-pivot-only}
Suppose the APSM transfer is used with $\eta=0$.
Then reporting $r_{i,t}=q_{i,t}$ is a weak expected-utility best
response.
The complete set of optimal reports is
\begin{equation}
\begin{aligned}
\mathcal R_{i,t}^{\star}
=
\bigl\{
r_i\in[0,1]:
a_t(r_i,r_{-i,t})
\in
\arg\max_{a\in\mathcal A_t}
\Phi_{i,t}(a;q_{i,t},r_{-i,t})
\bigr\}.
\end{aligned}
\label{eq:app-pivot-optimal-set}
\end{equation}
\end{proposition}

\begin{proof}
When $\eta=0$, Equation~\eqref{eq:app-groves-form} becomes
\begin{equation}
\overline U_{i,t}^{\mathrm{Pivot}}
=
\Phi_{i,t}
\!\left(
a_t(r_t);
q_{i,t},r_{-i,t}
\right)
-
H_{i,t}(r_{-i,t}).
\label{eq:app-pivot-only-utility}
\end{equation}
The second term is independent of $r_{i,t}$.

If agent $i$ reports $q_{i,t}$, the platform allocation rule
maximizes
\begin{equation}
x_{i,t}(a)q_{i,t}
+
\sum_{j\neq i}
x_{j,t}(a)r_{j,t},
\end{equation}
which is exactly
$\Phi_{i,t}(a;q_{i,t},r_{-i,t})$.
No alternative report can induce an allocation with a larger value
of $\Phi_{i,t}$.

Every report that induces an allocation maximizing $\Phi_{i,t}$ is
therefore utility-equivalent.
The set of such reports is precisely
$\mathcal R_{i,t}^{\star}$.
\end{proof}

\subsection{Proof of Theorem~\ref{thm:APSM-truthfulness}}
\label{app:proof-truthfulness}

\begin{proof}
Fix agent $i$, private belief $q_{i,t}$, and reports $r_{-i,t}$ of
the other agents.
Let
\begin{equation}
a_t^{\star}
\in
\arg\max_{a\in\mathcal A_t}
\Phi_{i,t}(a;q_{i,t},r_{-i,t}).
\label{eq:app-efficient-allocation}
\end{equation}

When agent $i$ reports $q_{i,t}$, the platform maximizes
\begin{equation}
x_{i,t}(a)q_{i,t}
+
\sum_{j\neq i}
x_{j,t}(a)r_{j,t},
\end{equation}
and therefore selects an allocation in the set defined by
Equation~\eqref{eq:app-efficient-allocation}.

For an arbitrary alternative report $r_{i,t}$,
Equation~\eqref{eq:app-groves-form} gives
\begin{equation}
\begin{aligned}
&
\overline U_{i,t}^{\mathrm{APSM}}
(q_{i,t},r_{-i,t})
-
\overline U_{i,t}^{\mathrm{APSM}}
(r_{i,t},r_{-i,t}) \\
&=
\Phi_{i,t}
\!\left(
a_t(q_{i,t},r_{-i,t});
q_{i,t},r_{-i,t}
\right) \\
&\quad-
\Phi_{i,t}
\!\left(
a_t(r_{i,t},r_{-i,t});
q_{i,t},r_{-i,t}
\right) \\
&\quad+
\eta
\left[
\overline S(q_{i,t};q_{i,t})
-
\overline S(r_{i,t};q_{i,t})
\right].
\end{aligned}
\label{eq:app-truthfulness-gap}
\end{equation}

The first difference is nonnegative because
$a_t(q_{i,t},r_{-i,t})$ maximizes $\Phi_{i,t}$.
By strong properness,
\begin{equation}
\overline S(q_{i,t};q_{i,t})
-
\overline S(r_{i,t};q_{i,t})
\ge
c_S(r_{i,t}-q_{i,t})^2.
\end{equation}
Thus,
\begin{equation}
\begin{aligned}
&
\overline U_{i,t}^{\mathrm{APSM}}
(q_{i,t},r_{-i,t})
-
\overline U_{i,t}^{\mathrm{APSM}}
(r_{i,t},r_{-i,t}) \\
&\qquad\ge
\eta c_S(r_{i,t}-q_{i,t})^2.
\end{aligned}
\label{eq:app-strict-utility-gap}
\end{equation}

If $r_{i,t}\neq q_{i,t}$ and $\eta>0$, the right-hand side is
strictly positive.
Therefore $r_{i,t}=q_{i,t}$ is the unique expected-utility best
response.
\end{proof}

\subsection{Proof of Corollary~\ref{cor:executor-error}}
\label{app:proof-executor}

\begin{proof}
By Theorem~\ref{thm:APSM-truthfulness},
\begin{equation}
\begin{aligned}
&
\overline U_{i,t}^{\mathrm{APSM}}
(q_{i,t},r_{-i,t})
-
\overline U_{i,t}^{\mathrm{APSM}}
(\widehat r_{i,t},r_{-i,t}) \\
&\qquad\ge
\eta c_S
(\widehat r_{i,t}-q_{i,t})^2.
\end{aligned}
\end{equation}
The executor suboptimality assumption implies that the left-hand
side is at most $\delta_{i,t}$.
Therefore
\begin{equation}
\eta c_S
(\widehat r_{i,t}-q_{i,t})^2
\le
\delta_{i,t}.
\end{equation}
Taking square roots gives
\begin{equation}
\left|
\widehat r_{i,t}-q_{i,t}
\right|
\le
\sqrt{
\frac{\delta_{i,t}}
{\eta c_S}
}.
\end{equation}
\end{proof}

\section{Additional Theoretical Analysis}
\label{app:additional-theory}

\subsection{Sparse-Audit Unbiasedness and Variance}
\label{app:sparse-variance}

Let
\begin{equation}
F_{i,t}
=
x_{i,t}\!\left(a_t(r_t)\right)y_{i,t}
+
\eta S(r_{i,t},y_{i,t})
\label{eq:app-outcome-component}
\end{equation}
denote the outcome-dependent component of the APSM transfer.

\begin{proposition}[Sparse-audit transfer]
\label{prop:app-sparse-audit}
Suppose $|F_{i,t}|\le M$ and audit sampling is independent of
$F_{i,t}$ conditional on the pre-round information.
Then
\begin{equation}
\mathbb E
\left[
\frac{z_{i,t}}{\epsilon_{i,t}}
F_{i,t}
\right]
=
\mathbb E[F_{i,t}],
\label{eq:app-ips-unbiased}
\end{equation}
and
\begin{equation}
\operatorname{Var}
\left(
\frac{z_{i,t}}{\epsilon_{i,t}}
F_{i,t}
\right)
=
\frac{\mathbb E[F_{i,t}^{2}]}{\epsilon_{i,t}}
-
\left(\mathbb E[F_{i,t}]\right)^2.
\label{eq:app-ips-exact-variance}
\end{equation}
Consequently,
\begin{equation}
\operatorname{Var}
\left(
\frac{z_{i,t}}{\epsilon_{i,t}}
F_{i,t}
\right)
\le
\frac{M^2}{\epsilon_{i,t}}.
\label{eq:app-ips-variance-bound}
\end{equation}
\end{proposition}

\begin{proof}
Conditional independence gives
\begin{equation}
\begin{aligned}
\mathbb E
\left[
\frac{z_{i,t}}{\epsilon_{i,t}}
F_{i,t}
\right]
&=
\mathbb E
\left[
\mathbb E
\left[
\frac{z_{i,t}}{\epsilon_{i,t}}
F_{i,t}
\mid
F_{i,t}
\right]
\right] \\
&=
\mathbb E
\left[
F_{i,t}
\frac{\mathbb E[z_{i,t}]}{\epsilon_{i,t}}
\right] \\
&=
\mathbb E[F_{i,t}].
\end{aligned}
\end{equation}

Moreover,
\begin{equation}
\begin{aligned}
\mathbb E
\left[
\left(
\frac{z_{i,t}}{\epsilon_{i,t}}
F_{i,t}
\right)^2
\right]
&=
\mathbb E
\left[
\frac{z_{i,t}}{\epsilon_{i,t}^2}
F_{i,t}^2
\right] \\
&=
\frac{\mathbb E[F_{i,t}^2]}{\epsilon_{i,t}}.
\end{aligned}
\end{equation}
Subtracting the squared mean yields
Equation~\eqref{eq:app-ips-exact-variance}.
The final bound follows from $F_{i,t}^2\le M^2$.
\end{proof}

Proposition~\ref{prop:app-sparse-audit} explains why changing the
audit rate does not alter the expected best response of a
risk-neutral agent-local executor, while lower audit rates increase
the variance and tail magnitude of realized transfers.

\subsection{Compensated Proper Scoring}
\label{app:compensated-brier}

Exact influence compensation and strictly proper scoring already
form a truthful analytical benchmark.

Consider the transfer
\begin{equation}
T_{i,t}^{\mathrm{CPS}}
=
-
G_{i,t}\!\left(a_t(r_t)\right)
+
\eta
\frac{z_{i,t}}{\epsilon_{i,t}}
S(r_{i,t},y_{i,t}).
\label{eq:app-compensated-brier}
\end{equation}
The resulting total utility is
\begin{equation}
U_{i,t}^{\mathrm{CPS}}
=
\eta
\frac{z_{i,t}}{\epsilon_{i,t}}
S(r_{i,t},y_{i,t}).
\end{equation}
Therefore,
\begin{equation}
\overline U_{i,t}^{\mathrm{CPS}}
=
\eta\overline S(r_{i,t};q_{i,t}),
\end{equation}
which is uniquely maximized at
$r_{i,t}=q_{i,t}$.

This benchmark shows that APSM is not the only truthful mechanism
under exact influence compensation and universally verifiable
reports.
The pivot component should instead be interpreted as a
Groves-consistent allocation and externality-pricing rule that
explicitly assigns the opportunity cost of occupying a limited
coordination position.

\subsection{Difference Rewards versus Pivot Charges}
\label{app:difference-reward}

A realized leave-one-out reward of the form
\begin{equation}
\Delta_{i,t}^{\mathrm{diff}}
=
W_t(a_t,y_t)
-
W_t(a_t^{(-i)},y_t)
\label{eq:app-difference-reward}
\end{equation}
is generally not a Clarke pivot charge.

First, Equation~\eqref{eq:app-difference-reward} compares realized
total welfare, whereas a pivot charge measures the externality
imposed on the other agents.
Second, the counterfactual $a_t^{(-i)}$ may delete agent $i$ from an
already chosen allocation without re-optimizing the remaining
positions.
Third, the difference reward is outcome-dependent, and therefore
may be sparse or delayed even when the allocation opportunity cost
is immediate.

In general,
\begin{equation}
G_i(a_t)
+
\Delta_{i,t}^{\mathrm{diff}}
\end{equation}
cannot be written as
\begin{equation}
\sum_j V_j(a_t)
-
h_i(r_{-i,t}),
\end{equation}
where $h_i$ is independent of $r_{i,t}$.
The standard Groves truthfulness argument therefore does not apply.

\subsection{Payoff Misspecification under Hard Top-K}
\label{app:payoff-misspecification}

Suppose the true selection benefit is
$G_i^{\mathrm{true}}(a)$, while the mechanism offsets
$\widehat G_i(a)$.
Define the residual benefit
\begin{equation}
B_i(a)
=
G_i^{\mathrm{true}}(a)
-
\widehat G_i(a).
\label{eq:app-residual-benefit}
\end{equation}

\begin{proposition}[Hard-selection non-robustness]
\label{prop:app-misspecification}
Suppose becoming selected increases the residual benefit by
$b_i>0$.
For every finite proper-score weight $\eta$ and every
$\gamma>0$, there exists a deterministic Top-$K$ report profile in
which the truthful report lies within $\gamma$ of the selection
threshold and an upward deviation is profitable.
\end{proposition}

\begin{proof}
It suffices to construct a Top-$1$ instance with two agents.
Let agent $i$ have private belief $q_i=q$, and let the competing
agent report
\begin{equation}
r_j=q+\frac{d}{2},
\end{equation}
for some $d>0$.

When agent $i$ reports truthfully, agent $j$ is selected.
If agent $i$ reports $q+d$, it becomes selected.
Relative to truthful allocation, the resulting loss in the
belief-space allocation objective is $d/2$.
The proper-score loss under the normalized Brier score is
$\eta d^2$, while the residual selection benefit is $b_i$.
The total utility gain from deviating is therefore
\begin{equation}
b_i
-
\frac{d}{2}
-
\eta d^2.
\label{eq:app-misspecification-gain}
\end{equation}

For every $b_i>0$, one can choose
\begin{equation}
0<d<\gamma
\end{equation}
sufficiently small that
\begin{equation}
b_i
>
\frac{d}{2}
+
\eta d^2.
\end{equation}
The deviation is then profitable.
\end{proof}

This result rules out a uniform approximate-truthfulness guarantee
for deterministic Top-$K$ allocation under arbitrary residual
selection benefits.
The main APSM theorem therefore assumes platform-controlled
selection benefits that can be accurately represented in the
transfer.

\subsection{Noisy Payoff Calculation}
\label{app:noisy-calculator}

Suppose the local executor uses an estimated utility
$\widehat{\mathcal U}_{i,t}(r)$ satisfying
\begin{equation}
\left|
\widehat{\mathcal U}_{i,t}(r)
-
\mathcal U_{i,t}(r)
\right|
\le
\xi_{i,t},
\qquad
\forall r\in\mathcal R.
\label{eq:app-uniform-utility-error}
\end{equation}

Let
\begin{equation}
r_{i,t}^{\star}
\in
\arg\max_{r\in\mathcal R}
\mathcal U_{i,t}(r),
\end{equation}
and
\begin{equation}
\widehat r_{i,t}
\in
\arg\max_{r\in\mathcal R}
\widehat{\mathcal U}_{i,t}(r).
\end{equation}

Then
\begin{equation}
\begin{aligned}
\mathcal U_{i,t}(r_{i,t}^{\star})
&\le
\widehat{\mathcal U}_{i,t}(r_{i,t}^{\star})
+\xi_{i,t} \\
&\le
\widehat{\mathcal U}_{i,t}(\widehat r_{i,t})
+\xi_{i,t} \\
&\le
\mathcal U_{i,t}(\widehat r_{i,t})
+
2\xi_{i,t}.
\end{aligned}
\end{equation}
Thus the executor is at most $2\xi_{i,t}$-suboptimal.
Corollary~\ref{cor:executor-error} gives
\begin{equation}
\left|
\widehat r_{i,t}
-
q_{i,t}
\right|
\le
\sqrt{
\frac{2\xi_{i,t}}
{\eta c_S}
}.
\label{eq:app-noisy-calculator-bound}
\end{equation}

\subsection{Finite Report Grids and Tie-Breaking}
\label{app:finite-grid}

Let $\mathcal R_h\subset[0,1]$ be a finite report grid with maximum
spacing $h$.
If $q_i\in\mathcal R_h$, Theorem~\ref{thm:APSM-truthfulness}
continues to imply that $q_i$ is the unique grid-optimal report.

If $q_i\notin\mathcal R_h$, let
\begin{equation}
\Pi_{\mathcal R_h}(q_i)
\in
\arg\min_{r\in\mathcal R_h}
|r-q_i|
\end{equation}
denote a nearest grid point.
Suppose that reporting $\Pi_{\mathcal R_h}(q_i)$ induces an
allocation maximizing
$\Phi_i(a;q_i,r_{-i})$.
Under the Brier score, every grid-optimal report then satisfies
\begin{equation}
\left|
r_i^{\star}-q_i
\right|
\le
\frac{h}{2}.
\label{eq:app-grid-error}
\end{equation}

If no nearest grid point preserves a belief-efficient allocation,
the grid can introduce an additional allocation-discretization
error.
Our implementation therefore reports both the report-grid spacing
and the truthful-allocation match rate.

For Pivot-only, multiple grid reports may induce the same efficient
allocation and therefore remain utility-equivalent.
We use a fixed, mechanism-independent tie-breaking rule when
reporting Pivot-only behavior.
The tie-breaking rule is never chosen based on distance to $q_i$.

\subsection{Surrogate-Welfare Approximation}
\label{app:surrogate-welfare}

The mechanism optimizes the additive expected-DCG surrogate
$\widetilde W_t$, while the downstream benchmark reports metrics
such as NDCG@5.

Let $W_t(a)$ denote a downstream allocation utility, and suppose
that
\begin{equation}
\left|
W_t(a)-\widetilde W_t(a)
\right|
\le
\zeta_t,
\qquad
\forall a\in\mathcal A_t.
\label{eq:app-surrogate-error}
\end{equation}
Let
\begin{equation}
a_t^{\star}
\in
\arg\max_{a\in\mathcal A_t}
W_t(a),
\qquad
\widetilde a_t
\in
\arg\max_{a\in\mathcal A_t}
\widetilde W_t(a).
\end{equation}
Then
\begin{equation}
W_t(a_t^{\star})
-
W_t(\widetilde a_t)
\le
2\zeta_t.
\label{eq:app-surrogate-regret}
\end{equation}

Indeed,
\begin{equation}
\begin{aligned}
W_t(a_t^{\star})
-
W_t(\widetilde a_t)
&\le
W_t(a_t^{\star})
-
\widetilde W_t(a_t^{\star}) \\
&\quad+
\widetilde W_t(\widetilde a_t)
-
W_t(\widetilde a_t) \\
&\le
2\zeta_t,
\end{aligned}
\end{equation}
where we use the optimality of $\widetilde a_t$ for the surrogate
objective.

\section{Efficient Pivot Computation}
\label{app:efficient-pivot}

For the Top-$K$ position-allocation instance, all pivot charges can
be computed after a single sorting operation.

Let
\begin{equation}
r_{(1)}
\ge
r_{(2)}
\ge
\dots
\ge
r_{(N)}
\end{equation}
denote reports sorted according to the platform's fixed
tie-breaking rule.
The top $K$ agents receive position values
$d_1,\dots,d_K$.
Define
\begin{equation}
d_{K+1}=0.
\end{equation}

For the agent placed at rank $k\le K$, removing the agent shifts
the candidates at ranks $k+1,\dots,K+1$ upward.
Its pivot charge is
\begin{equation}
p_{(k)}^{\mathrm{pivot}}
=
\sum_{\ell=k}^{K}
\left(
d_{\ell}-d_{\ell+1}
\right)
r_{(\ell+1)}.
\label{eq:app-topk-pivot}
\end{equation}
If fewer than $K+1$ agents are available, missing outside reports
are defined as zero.
Unselected agents impose no allocation externality and receive
\begin{equation}
p_{(k)}^{\mathrm{pivot}}=0,
\qquad
k>K.
\end{equation}

Equation~\eqref{eq:app-topk-pivot} can be evaluated for all selected
agents using suffix sums after sorting.
The overall computational complexity per round is therefore
\begin{equation}
O(N\log N+K),
\end{equation}
and the additional memory cost is $O(N)$.

\section{Experimental Protocol and Additional Results}
\label{app:experiments}

This appendix provides the complete protocol for the controlled
mechanism replay and reports analyses omitted from the main paper for
space. We maintain the same distinction as the main text: $q_i$ is a
frozen agent-local reference belief formed before mechanism-dependent
reporting, whereas $r_i$ is the report submitted to the platform. The
reference belief is used only inside the agent-local policy executor and
for evaluation; it is never exposed to the platform allocation or
transfer modules.

\subsection{Experimental Protocol}
\label{app:protocol}

\paragraph{Paired controlled replay.}
We separate candidate generation and belief formation from strategic
reporting. For every query or database incident, we first freeze the
candidate content, supporting evidence, realized outcomes, agent-local
reference beliefs, admissible report domain, allocation tie-breaking,
and random seeds. All mechanisms start from the same cached initial
profile and use the same agent-update order and sweep schedule. Within a
mechanism, reports may subsequently differ because the public utility
function differs. Consequently, paired differences are not confounded by
changes in generated content, task difficulty, initialization, or audit
randomness.

\paragraph{AgentWebBench content allocation.}
The content-allocation environment contains six agents per query: two Qwen3-8B agents, two Qwen3-4B agents, and two DeepSeek V4 pro agents competing in the same market.
The platform selects $K=2$ content agents using the
additive surrogate defined in the main paper. Each selected
agent contributes up to three documents. The documents are
ordered according to the selected agents' reports, merged,
deduplicated, and truncated to the top five for
document-level NDCG@5 evaluation. 

\paragraph{Database root-cause diagnosis.}
The diagnosis environment is constructed from the Database Root-Cause
Diagnosis task of MultiAgentBench. It contains 50 independently
initialized PostgreSQL incidents covering ten single- and dual-root-cause
configurations, with five incidents per configuration. Six specialist
agents execute hypothesis-specific SQL diagnostics and exchange findings
over two communication rounds. Each incident is repeated under three
paired rotations of the specialist--model assignment, using two Qwen3-8B agents, two Qwen3-4B agents, and two DeepSeek V4 pro agents. All mechanisms use the same Top-$2$
router and the same frozen diagnoses, evidence, reference beliefs,
outcomes, initialization, and tie-breaking.

\paragraph{Public mechanism configuration.}
All compared mechanisms use the same public allocation values $x_i(a)$,
the same platform-controlled selection-benefit schedule $G_i(a)$, and
the same normalized Brier score
\begin{equation}
S(r,y)=1-(r-y)^2.
\label{eq:app_brier_score}
\end{equation}
We set the proper-score weight to $\eta=1$ and use a default audit
probability of $\epsilon=0.3$ unless otherwise stated. The ranked
position values and selection-benefit schedule are fixed before reports
are submitted and are reused without modification across mechanisms.
The main controlled replay maximizes over a cached finite admissible
report domain; the separate finite-grid ablation replaces that domain
with a uniform grid of spacing $0.05$. Therefore, ``exact execution'' in
an experiment means exact maximization over the stated admissible domain,
not necessarily over the continuum $[0,1]$. Exact recovery of $q_i$
requires that $q_i$ be admissible. Allocation ties are resolved by smaller agent index.
Utility ties are resolved by selecting the smallest report among
maximizing candidates (a fixed, mechanism-independent rule).
For Pivot-only, we retain and evaluate the complete maximizing set
rather than forcing a unique report.

\paragraph{Mechanism-agnostic agent-local executor.}
For the controlled replay, the executor runs inside the agent-local trust
boundary. Given private reference belief $q_i$, the current public
reports $r_{-i}$, and an admissible report domain $\mathcal R_i$, it
evaluates each candidate report by
\begin{equation}
\overline U_i(r;r_{-i},q_i)
=
\mathbb E_{Y\sim\mathrm{Bernoulli}(q_i),Z}
\left[U_i(r,r_{-i},Y,Z)\right]
\label{eq:app_expected_utility}
\end{equation}
and submits
\begin{equation}
\widehat r_i
\in
\arg\max_{r\in\mathcal R_i}
\overline U_i(r;r_{-i},q_i).
\label{eq:app_executor}
\end{equation}
For every candidate report, the executor recomputes the induced ranked
allocation, evaluates the mechanism-specific transfer, and calculates
the agent's expected utility. The same code path, report domain,
initialization, update order, and tie-breaking are used for every
mechanism. In particular, the executor contains no rule that directly
sets $r_i=q_i$. Only the selected report $\widehat r_i$ is released to
the platform.

Each agent optimizes independently against the frozen reference beliefs
$q_{-i}$ of all other agents (unilateral, simultaneous best response).
The optimal reports are computed independently and then combined into a
joint profile for Top-$K$ allocation. No agent observes another agent's
mechanism-specific report before submitting its own. The same
initialization is used in every paired run; only the utility function
differs across mechanisms.


\paragraph{Calculator-free LLM reporting.}
The calculator-free experiment in RQ1 does not use
Equation~\eqref{eq:app_executor}. We freeze each agent's candidate,
evidence, and initial confidence, and ask the same LLM to submit a final
confidence under three public conditions: neutral re-report,
Selection-only, and Matched-Brier. The LLM retains control of the final
reported probability in every condition. The neutral re-report condition
captures ordinary variation caused by asking for confidence a second
time; incentive-induced shifts are therefore measured relative to this
paired control. Reports that cannot be parsed as a probability in
$[0,1]$ are counted as parse errors and are not silently repaired. The
eligible denominator for the threshold-crossing rate is reported in the
main-paper table caption.

\paragraph{Tool-assisted LLM reporting.}
For each frozen instance, the LLM receives the candidate,
supporting evidence, reference confidence, competing reports,
admissible report domain, and public mechanism description.
It may invoke an agent-local utility tool that returns the
expected utility of every admissible report. The tool neither
marks a report as optimal nor submits a report. The LLM
independently produces its final confidence after inspecting
the returned table.

Matched-Brier, CPS, and APSM use identical instances,
model--agent assignments, report domains, prompt structure,
tool access, parsing rules, and query-level pairing. Only the
mechanism-specific utility calculation differs. Tool invocation
succeeds in $100\%$ of runs for all three mechanisms.

\paragraph{Compared mechanisms.}
All mechanisms use the same frozen candidates, initial report profile,
allocation rule, admissible report domain, update schedule, and audit
random-number stream.
\emph{Selection-only} provides only the platform-controlled allocation
benefit $G_i(a(r))$.
\emph{Matched-Brier} adds the same sparsely audited proper score but does
not compensate the allocation benefit.
\emph{Difference Reward} (LOO in the appendix figures) uses the realized
leave-one-out contribution reward implemented in the evaluation code;
unlike the Clarke charge, this baseline is outcome-based rather than a
reported-welfare externality payment.
\emph{Compensated Proper Scoring} (CPS) offsets $G_i$ and applies the same
proper score,
\begin{equation}
T_i^{\mathrm{CPS}}
=
-G_i(a(r))
+
\eta\frac{z_i}{\epsilon_i}S(r_i,y_i).
\label{eq:app_cps}
\end{equation}
\emph{Pivot-only} (Clarke in the appendix figures) uses the complete APSM
transfer with $\eta=0$, retaining allocation-benefit compensation,
audited realized value, and the Clarke pivot charge while omitting strict
probability identification.
\emph{APSM} uses the complete transfer from the main paper.
The \emph{reference-belief oracle} ranks candidates directly using $q$
and is used only as a non-deployable evaluation reference.


\subsection{Metrics and Statistical Analysis}
\label{app:metrics}

\paragraph{Reporting fidelity.}
The primary reporting metric is mean absolute deviation from the frozen
agent-local reference belief,
\begin{equation}
\mathrm{RefDev}
=
\frac{1}{|\mathcal D|N}
\sum_{t\in\mathcal D}\sum_{i=1}^{N}
|r_{i,t}-q_{i,t}|.
\label{eq:app_refdev}
\end{equation}
For calculator-free reporting, the signed incentive shift is measured
relative to the paired neutral re-report condition. We additionally
report the threshold-crossing rate among agents whose frozen reference
report is below the Top-$K$ cutoff, and the exact reference-report rate
when appropriate.

\paragraph{Probability and allocation quality.}
All-candidate probability quality is measured over the frozen candidate
set,
\begin{equation}
\mathrm{Brier}
=
\frac{1}{|\mathcal D|N}
\sum_{t\in\mathcal D}\sum_{i=1}^{N}
(r_{i,t}-y_{i,t})^2.
\label{eq:app_brier_metric}
\end{equation}
Let $a_t(r_t)$ and $a_t(q_t)$ be the ranked allocations induced by reports
and reference beliefs. Exact reference-allocation match is
\begin{equation}
\mathrm{Match}
=
\frac{1}{|\mathcal D|}
\sum_{t\in\mathcal D}
\mathbb I\!\left[a_t(r_t)=a_t(q_t)\right].
\label{eq:app_match}
\end{equation}
AgentWebBench additionally reports NDCG@5. Database diagnosis reports
Top-$1$ success, the benchmark-defined Top-$2$ success criterion, and mean
reciprocal rank (MRR).

\paragraph{Allocation-changing utility margin.}
For comparisons between CPS and APSM, define the set of admissible reports
that change the reference-induced allocation as
\begin{equation}
\mathcal A_i^{\mathrm{chg}}
=
\left\{r_i\in\mathcal R_i:
 a(r_i,r_{-i})\neq a(q_i,r_{-i})\right\}.
\label{eq:app_changing_set}
\end{equation}
For agent--instance pairs with $\mathcal A_i^{\mathrm{chg}}\neq\emptyset$,
the allocation-changing margin is
\begin{equation}
M_i^{\mathrm{alloc}}
=
\overline U_i(q_i,r_{-i})
-
\max_{r_i\in\mathcal A_i^{\mathrm{chg}}}
\overline U_i(r_i,r_{-i}).
\label{eq:app_allocation_margin}
\end{equation}
A larger value indicates that an approximate executor must make a larger
utility error before selecting an allocation-changing report. Pairs with
no admissible allocation-changing report are not included in this
diagnostic because the constrained maximum is undefined.


\paragraph{Paired uncertainty estimates.}
For the calculator-free and tool-assisted RQ1 experiments,
confidence intervals are computed by block bootstrapping the
200 queries. Controlled AgentWebBench results
use 10,000 seed--query hierarchical bootstrap draws: seeds are resampled
first, followed by queries within each selected seed. For database
diagnosis, the 50 independently initialized incidents are treated as
statistical clusters. Paired differences are first averaged over the
three specialist--model rotations and then resampled with 5,000
incident-cluster bootstrap draws. Agent-level observations from the same
query or incident are never treated as independent samples. Error bars
and shaded regions denote 95\% confidence intervals unless otherwise
specified.

\subsection{Complete AgentWebBench Results}
\label{app:agentweb_complete}

\begin{table*}[t]
    \centering
    \small
    \setlength{\tabcolsep}{5pt}
    \caption{Complete AgentWebBench results over 200 queries and three paired
    seeds. The population contains two Qwen3-8B, two Qwen3-4B, and two
    DeepSeek agents. Tail failure is the fraction of queries with no relevant
    selected candidate.}
    \label{tab:app_agentweb_complete}
    \begin{tabular}{lcccccc}
        \hline
        Mechanism & $|r-q|\downarrow$ & Brier$\downarrow$ & NDCG@5$\uparrow$
        & Sel. eff.$\uparrow$ & Tail fail.$\downarrow$ & Exact match$\uparrow$ \\
        \hline
        Selection & 0.361 & 0.781 & 0.274 & 0.758 & 0.650 & 0.022 \\
        Matched-Brier & 0.354 & 0.770 & 0.328 & 0.800 & 0.600 & 0.000 \\
        Pivot-only & 0.059 & 0.472 & 0.379 & 0.880 & 0.520 & 0.702 \\
        APSM & \textbf{0.001} & \textbf{0.415} & \textbf{0.452}
        & \textbf{0.938} & \textbf{0.462} & \textbf{0.992} \\
        \hline
    \end{tabular}
\end{table*}
Table~\ref{tab:app_agentweb_complete} reports the complete
AgentWebBench results. Matched-Brier retains substantial
strategic deviation ($0.354$) and never reproduces the
reference-induced allocation. Pivot-only reduces the deviation
to $0.059$ and improves allocation match to $0.702$, but does
not uniquely identify the reference probability. APSM reduces
the deviation to $0.001$, achieves a $0.992$ exact allocation
match, and obtains the highest NDCG@5 ($0.452$). The same
ordering is observed for selection efficiency and tail failure.

\begin{figure*}[t]
    \centering
    \includegraphics[width=0.98\textwidth]{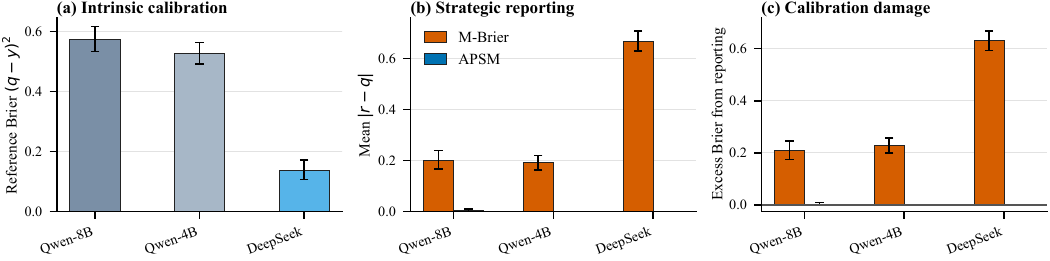}
    \caption{Model-level decomposition in the mixed 2--2--2 AgentWebBench
    population. \textbf{(a)} Intrinsic Brier loss of the frozen reference
    beliefs. \textbf{(b)} Mean absolute reporting deviation under
    Matched-Brier and APSM. \textbf{(c)} Excess Brier loss introduced by
    strategic reporting. Error bars are 95\% query-level bootstrap confidence
    intervals over 200 unique queries.}
    \label{fig:app_model_level}
\end{figure*}

Figure~\ref{fig:app_model_level} decomposes the results across
the mixed 2--2--2 population. DeepSeek produces the lowest
intrinsic Brier loss ($0.138$), compared with $0.575$ for
Qwen3-8B and $0.528$ for Qwen3-4B. Nevertheless, under
Matched-Brier, DeepSeek exhibits the largest reporting
deviation ($0.668$), compared with $0.202$ and $0.191$ for
the two Qwen families. Its reporting-induced Brier increase is
also substantially larger.

APSM reduces reporting deviation to $0.004$ for Qwen3-8B
and approximately zero for Qwen3-4B and DeepSeek, while
nearly eliminating the corresponding calibration damage.
This result shows that better passive calibration does not by
itself imply incentive-compatible elicitation. It should not be
interpreted as evidence that one model family is inherently
more strategic, since the observed deviation also depends on
the reference belief's distance from the Top-$K$ threshold.

\subsection{Collaborative Calibration Followed by Strategic Reporting}
\label{app:cc}

Collaborative Calibration (CC) acts on belief formation, whereas APSM acts
on the subsequent reporting layer. We therefore first freeze the
original and CC-revised belief profiles and then apply different reporting
mechanisms to the same CC reference beliefs.

\begin{figure*}[t]
    \centering
    \includegraphics[width=0.98\textwidth]{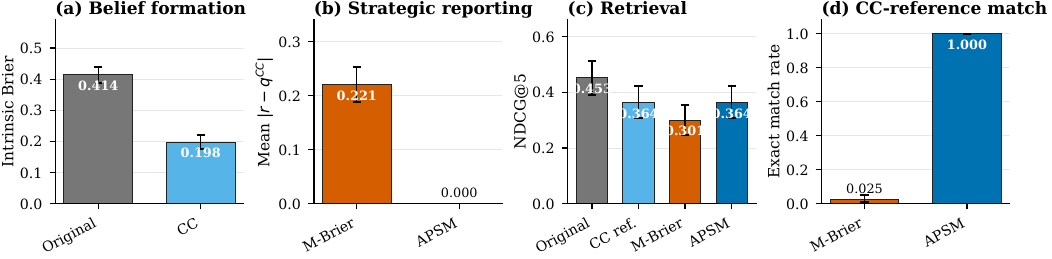}
    \caption{Interaction between Collaborative Calibration (CC) and
    mechanism-dependent reporting. \textbf{(a)} Intrinsic Brier loss of
    the original and CC-revised reference beliefs. \textbf{(b)} Mean
    absolute deviation from the frozen CC reference beliefs after
    Matched-Brier or APSM reporting. \textbf{(c)} NDCG@5 induced by the
    original reference, the CC reference, and the two reports formed from
    the CC reference. \textbf{(d)} Exact match to the allocation induced
    by the CC reference beliefs. Results use 200 queries, three paired
    seeds, and six agents per query.}
    \label{fig:app_cc}
\end{figure*}

As shown in Figure~\ref{fig:app_cc}(a), CC reduces intrinsic
Brier loss from $0.414$ to $0.198$, confirming that
collaboration can improve the marginal probability quality of
the beliefs supplied to the mechanism. The induced ranking is
not necessarily improved: NDCG@5 decreases from $0.453$ for
the original reference profile to $0.364$ for the CC reference
profile. This difference reflects the distinction between
marginal calibration and cross-agent resolution under a scarce
ranked allocation.

Holding the CC beliefs fixed, Matched-Brier introduces a mean
absolute reporting deviation of $0.221$, decreases NDCG@5 to
$0.301$, and reproduces the CC-reference allocation in only
$2.5\%$ of queries. APSM returns the CC reference reports in
the controlled replay, obtains zero reference deviation,
restores NDCG@5 to $0.364$, and achieves an exact
allocation-match rate of $1.000$. APSM therefore does not
replace collaborative calibration; it prevents the downstream
selection incentive from introducing an additional
$q_i^{\mathrm{CC}}\!\rightarrow r_i$ distortion.

\subsection{Complete Database-Diagnosis Comparison}
\label{app:database_complete}

\begin{figure*}[t]
    \centering
    \includegraphics[width=0.98\textwidth]{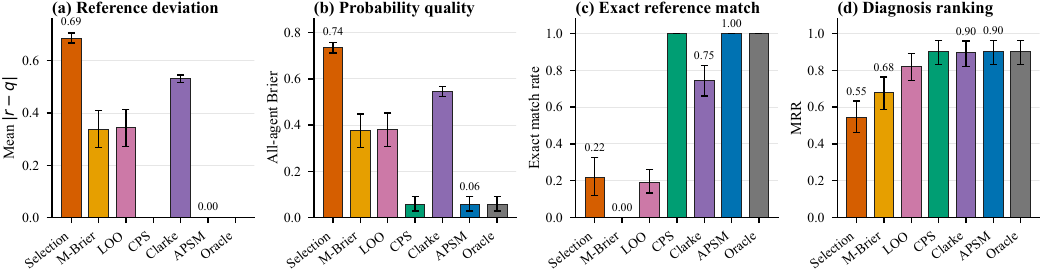}
    \caption{Complete mechanism comparison on MultiAgentBench database
    root-cause diagnosis. \textbf{(a)} Mean absolute reference
    deviation. \textbf{(b)} All-agent Brier loss. \textbf{(c)} Exact
    reference-allocation match. \textbf{(d)} Diagnosis ranking measured
    by MRR. LOO denotes Difference Reward, and Clarke denotes Pivot-only.
    CPS and APSM share exact influence compensation and the same strictly
    proper score; the oracle directly allocates using the frozen reference
    beliefs. Error bars use paired incident-cluster bootstrap confidence
    intervals.}
    \label{fig:app_database_complete}
\end{figure*}

Figure~\ref{fig:app_database_complete} extends the main-paper table by
including CPS and the reference-belief oracle. Selection-only has the
largest distortion (approximately $0.69$) and the weakest diagnosis
ranking (MRR $0.55$). Matched-Brier and Difference Reward reduce the
error but retain substantial reference deviation. Pivot-only reaches an
MRR of approximately $0.90$ and an exact allocation-match rate of
approximately $0.75$, but its mean reference deviation remains
approximately $0.53$ because it identifies an efficient allocation
rather than a unique probability report.

Under exact maximization over the evaluated admissible domain, CPS and
APSM both obtain zero reference deviation, match the reference-induced
Top-$2$ allocation on every incident, and reach an MRR of approximately
$0.90$, equal to the reference-belief oracle at the displayed precision.
This confirms that exact compensation plus a strictly proper score is
already sufficient for exact probability elicitation whenever the
reference report is admissible. The additional audited Groves allocation
component in APSM is evaluated through utility margins and approximate
execution rather than through a claim that CPS is non-truthful under
exact optimization.

\begin{table}[t]
    \centering
    \small
    \setlength{\tabcolsep}{3.3pt}
    \caption{Database diagnosis results separated into single- and dual-root
    incidents. Each value averages three paired specialist--model rotations.}
    \label{tab:app_database_breakdown}
    \begin{tabular}{clccccc}
        \hline
        Roots & Mechanism & $|r-q|$ & Brier & Top-1 & MRR & Match \\
        \hline
        1 & M-Brier & 0.127 & 0.194 & 0.653 & 0.747 & 0.000 \\
        1 & APSM & \textbf{0.000} & \textbf{0.104} & \textbf{0.707}
          & \textbf{0.806} & \textbf{1.000} \\
        2 & M-Brier & 0.551 & 0.563 & 0.480 & 0.609 & 0.000 \\
        2 & APSM & \textbf{0.000} & \textbf{0.010} & \textbf{1.000}
          & \textbf{1.000} & \textbf{1.000} \\
        \hline
    \end{tabular}
\end{table}
Table~\ref{tab:app_database_breakdown} separates single- and
dual-root-cause incidents. The gap is largest for dual-root
incidents, where Matched-Brier produces a reference deviation
of $0.551$ and an MRR of $0.609$. APSM eliminates the
reporting deviation and reaches an MRR of $1.000$. For
single-root incidents, APSM likewise improves MRR from
$0.747$ to $0.806$ while exactly preserving the
reference-induced allocation.

\subsection{Mechanism Decomposition}
\label{app:decomposition}

\begin{figure*}[t]
    \centering
    \includegraphics[width=0.96\textwidth]{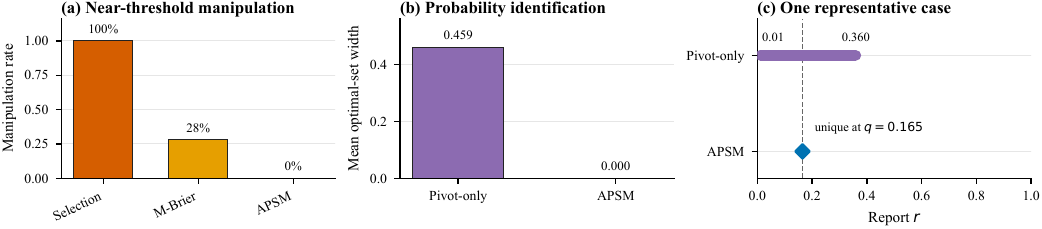}
   \caption{Mechanism decomposition. \textbf{(a)} Near-threshold
manipulation rate. \textbf{(b)} Mean optimal-report-set width.
\textbf{(c)} Pivot-only admits $[0.01,0.360]$, whereas APSM
uniquely identifies $q=0.165$.}
    \label{fig:app_decomposition}
\end{figure*}

For the near-threshold diagnostic, we hold each controlled instance fixed
and identify agents whose reference report lies below a deterministic
Top-$K$ cutoff but whose admissible report domain contains a cutoff-crossing
alternative. An instance is counted as manipulated when the
utility-maximizing report crosses the cutoff. The manipulation rate is the
fraction of these eligible instances that are manipulated.

Figure~\ref{fig:app_decomposition}(a) isolates the discontinuity created
by deterministic Top-$K$ allocation. Selection-only produces a $100\%$
manipulation rate. A scale-matched Brier score lowers this rate to $28\%$,
but cannot eliminate profitable crossings because the proper-score loss
is locally quadratic while the selection gain is discrete. APSM obtains
a $0\%$ manipulation rate in the same controlled instances when the
platform-controlled benefit is correctly compensated.

Panels (b)--(c) isolate the role of strict probability
scoring. The audited Groves allocation term in Pivot-only
makes an efficient allocation a weak best response, but every
report inducing the same efficient allocation may remain
optimal. Across the controlled instances, the mean width of
this optimal-report set is $0.459$. In the representative case
shown in panel (c), Pivot-only admits all reports in
$[0.01,0.360]$, whereas the strictly proper term in APSM
uniquely identifies $q=0.165$. Thus, the allocation component
determines which allocation is desirable, whereas the proper
score uniquely identifies the probability report.

\subsection{Sparse-Audit Verification--Variance Trade-off}
\label{app:sparse_audit}

To isolate settlement variability from changes in strategic
behavior, we use paired audit replay. Reports and allocations
are first selected by maximizing expected utility and then
held fixed. We replay matched audit draws under
\begin{equation}
\epsilon\in\{0.05,0.1,0.2,0.3,0.5,0.7,1.0\}
\label{eq:app_audit_rates}
\end{equation}
and recompute the corresponding inverse-propensity
transfers. This protocol isolates the effect of audit sparsity
without assuming that an unaided LLM relearns its reporting
policy at every audit rate.

For each rate, we report agent-level transfer variance, the
mean per-instance maximum absolute transfer, and the fraction
of instances in which no agent is audited.

\begin{figure*}[t]
    \centering
    \includegraphics[width=0.98\textwidth]{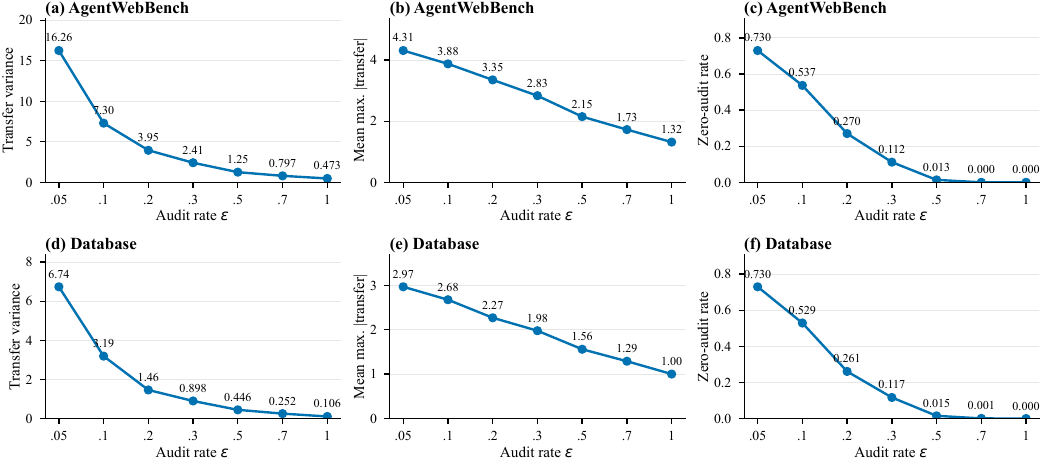}
    \caption{Sparse-audit trade-off on AgentWebBench and
    database diagnosis for seven audit rates.
    \textbf{(a,d)} Variance of realized agent transfers.
    \textbf{(b,e)} Mean per-instance maximum absolute
    transfer. \textbf{(c,f)} Fraction of instances containing
    no audited agent. Reports, allocations, and matched
    random-number streams are fixed across audit rates.}
    \label{fig:app_sparse_audit}
\end{figure*}

Figure~\ref{fig:app_sparse_audit} shows the same smooth trend
in both environments. As auditing becomes sparser, transfer
variance and payment tails increase sharply because
inverse-propensity weighting amplifies each audited outcome.
The zero-audit rate also rises rapidly at low audit
probabilities, whereas it is nearly eliminated once the audit
rate reaches approximately $0.7$. The curves therefore reveal a clear verification--variance
trade-off. Moderate audit sparsity reduces verification cost
while retaining relatively stable settlement, but aggressively
sparse auditing produces high-variance and occasionally
uninformative rounds. This behavior is consistent with the
theoretical $O(1/\epsilon)$ dependence and highlights that
unbiased expected incentives do not guarantee low realized
payment risk.

\subsection{Sensitivity to Selection-Benefit Misspecification}
\label{app:misspecification}

The truthfulness result assumes that the platform correctly models the
allocation benefit it controls. We test a multiplicative mismatch by
letting the agent's true benefit be
\begin{equation}
G_i^{\mathrm{true}}(a)
=
c\,\widehat G_i(a),
\label{eq:app_misspec}
\end{equation}
while the mechanism continues to compensate only $\widehat G_i(a)$. We
sweep
\begin{equation}
c\in\{0.90,0.95,0.98,1.00,1.02,1.05,1.10\}.
\label{eq:app_misspec_sweep}
\end{equation}
All candidates, beliefs, allocation rules, report domains, proper-score
weights, initial profiles, and paired executor randomness remain fixed.

\begin{figure*}[t]
    \centering
    \includegraphics[width=0.98\textwidth]
    {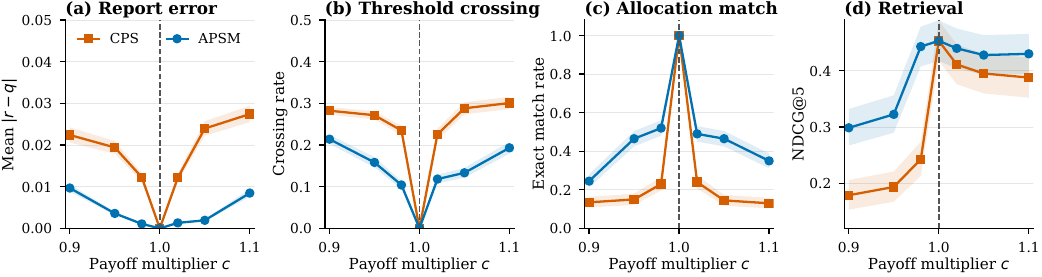}
    \caption{Sensitivity to misspecification of the platform-controlled
    allocation benefit. The agent's true benefit is $c\widehat G_i$,
    whereas CPS and APSM compensate $\widehat G_i$. \textbf{(a)} Mean
    absolute report error. \textbf{(b)} Fraction of reports that cross
    the reference-induced allocation threshold. \textbf{(c)} Exact
    reference-allocation match. \textbf{(d)} NDCG@5. The vertical dashed
    line marks exact specification ($c=1$). Shaded regions denote 95\%
    paired confidence intervals.}
    \label{fig:app_misspecification}
\end{figure*}

At $c=1$, CPS and APSM both recover the reference report and allocation
under exact execution. Any mismatch leaves a residual selection benefit
and sharply reduces exact allocation recovery near hard Top-$K$
thresholds. APSM nevertheless degrades more slowly than CPS because an
allocation-changing report must overcome both the proper-score loss and
the belief-space allocation loss in the audited Groves term. At a
representative $\pm5\%$ mismatch, APSM retains a $46.5\%$ exact
reference-allocation match rate, compared with approximately $15\%$ for
CPS. APSM also produces lower report error and fewer threshold crossings
throughout the displayed sweep, while maintaining higher NDCG@5 away
from exact compensation.

This experiment demonstrates relative resilience rather than an absolute
robustness guarantee. Both mechanisms decline substantially from their
$100\%$ exact match at $c=1$, consistent with the sensitivity of
deterministic Top-$K$ allocation to residual selection benefits.

\subsection{Approximate Agent-Local Execution}
\label{app:approx_execution}

We evaluate six controlled executor conditions. \emph{Grid} performs
exact utility maximization on a uniform report grid with spacing $0.05$.
\emph{N .01} and \emph{N .05} perturb candidate-report utilities with
paired additive-noise realizations of magnitude $0.01$ and $0.05$,
respectively. \emph{MC10}, \emph{MC50}, and \emph{MC200} estimate
expected utility using 10, 50, and 200 Monte Carlo samples of the random
variables appearing in the mechanism utility. CPS and APSM reuse the same
perturbation realizations or Monte Carlo random-number streams in each
paired instance.

For a submitted report $\widehat r_i$, define its expected-utility
suboptimality relative to the truthful optimum as
\begin{equation}
\delta_i
=
\overline U_i(q_i,r_{-i})
-
\overline U_i(\widehat r_i,r_{-i}).
\label{eq:app_delta}
\end{equation}
For APSM, Corollary~1 implies
\begin{equation}
|\widehat r_i-q_i|
\le
\sqrt{\frac{\delta_i}{\eta c_S}}.
\label{eq:app_execution_bound}
\end{equation}
The same bound holds for CPS under exact compensation because, with the
normalized Brier score,
$\overline U_i^{\mathrm{CPS}}(q_i)-
\overline U_i^{\mathrm{CPS}}(r_i)=\eta(r_i-q_i)^2$.
With $\eta=1$ and $c_S=1$, the displayed bound becomes
$\sqrt{\delta_i}$.

\begin{figure*}[t]
    \centering
    \includegraphics[width=0.98\textwidth]
    {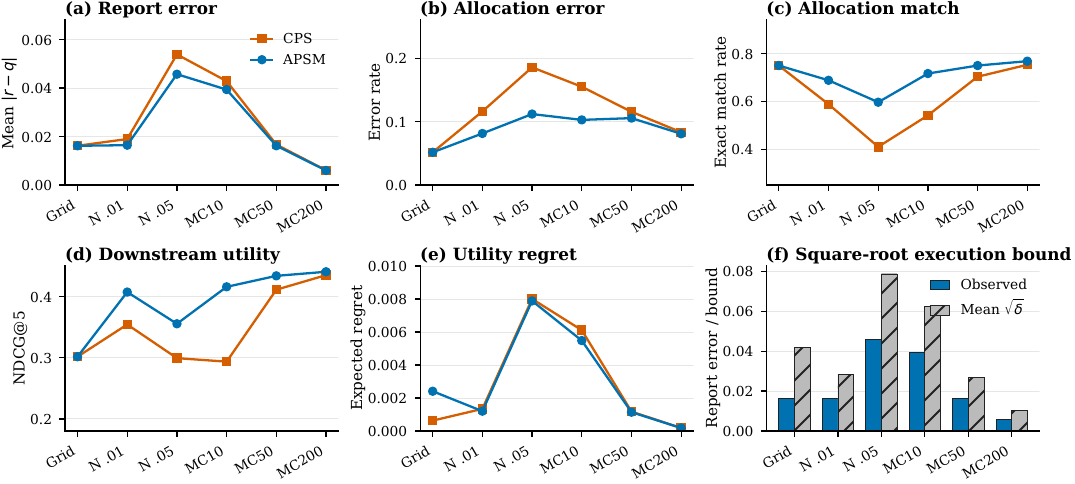}
    \caption{CPS and APSM under approximate agent-local execution.
    \textbf{(a)} Mean absolute report error. \textbf{(b)} Fraction of
    reports that change the reference-induced allocation. \textbf{(c)}
    Exact reference-allocation match. \textbf{(d)} Downstream NDCG@5.
    \textbf{(e)} Expected-utility suboptimality $\delta$. \textbf{(f)}
    Observed APSM report error compared with the mean square-root
    execution bound from Equation~\eqref{eq:app_execution_bound}.
    Grid uses report spacing $0.05$; N .01 and N .05 use additive
    utility-noise magnitudes $0.01$ and $0.05$; MC10, MC50, and MC200
    use 10, 50, and 200 utility samples. All paired conditions reuse
    the same frozen candidates, beliefs, allocation rule, perturbation
    realizations, and Monte Carlo streams.}
    \label{fig:app_executor_approximation}
\end{figure*}

Report error alone is similar for CPS and APSM in several conditions,
but their allocation consequences differ. Under utility noise of $0.05$,
CPS changes the reference-induced allocation in $18.6\%$ of instances,
whereas APSM reduces this rate to $11.2\%$. The corresponding NDCG@5
increases from $0.299$ under CPS to $0.356$ under APSM, while exact
allocation match increases from $41.0\%$ to $59.7\%$.

This difference follows from the exact-execution allocation margin. The
mean margin is $0.3350$ for CPS and $0.6462$ for APSM, corresponding to
a $1.93\times$ increase. An approximate executor can therefore make a
numerically similar reporting error under both mechanisms while being
less likely to cross an allocation boundary under APSM. Finally,
Figure~\ref{fig:app_executor_approximation}(f) shows that the observed
APSM report error remains below the empirical mean square-root bound in
every condition, consistent with the approximate-execution result.

\paragraph{Implementation and reproducibility.}
The experiments were conducted on an Ubuntu 22.04.5 server with dual
Intel Xeon Platinum 8368Q CPUs, 503~GB of system memory, and eight
NVIDIA RTX 4090 GPUs. Local Qwen inference was served through an
OpenAI-compatible vLLM interface with a maximum context length of 4096.

For every paired experiment, we cache the query or incident identifier,
model--agent assignment, generated candidate content, supporting
evidence, reference belief, realized outcome, admissible report domain,
allocation tie-breaking order, initial report profile, update order, and
audit random-number stream. The accompanying experiment configuration
additionally specifies the public position-value vector, the exact
selection-benefit schedule and scale, all mechanism parameters, the
best-response stopping criterion, full prompts, parse-error handling,
noise distributions, the Monte Carlo sampling protocol, and the
bootstrap scripts used to compute confidence intervals. Together, these
artifacts fully determine the reported mechanism comparisons and paired
uncertainty estimates without requiring parameter values to be inferred
from the figures.

\subsection{CPS Score-Weight Sweep}
\label{app:cps_eta_sweep}

We further test whether increasing the proper-score weight
$\eta$ can reproduce APSM's advantage under approximate
execution. We vary
\[
\eta\in\{0.5,1,1.43,2,4,8\},
\]
where $\eta=1.43$ matches APSM's mean per-instance utility
range. All conditions use the same frozen candidates, reference
beliefs, report set, allocation rule, and paired additive utility
noise with $\sigma=0.05$.

\begin{figure*}[t]
    \centering
    \includegraphics[width=0.98\textwidth]
    {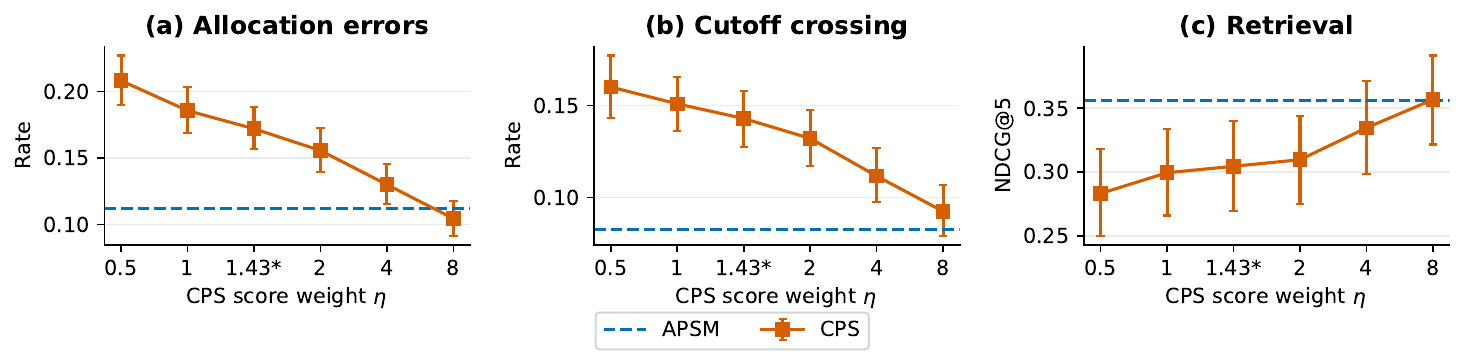}
    \caption{CPS score-weight sweep under additive utility
    noise $\sigma=0.05$. \textbf{(a)} Allocation-changing
    error rate. \textbf{(b)} Cutoff-crossing rate.
    \textbf{(c)} NDCG@5. Dashed lines show APSM with
    $\eta=1$; the starred point matches APSM's mean utility
    range. Error bars denote 95\% seed--query hierarchical
    bootstrap intervals.}
    \label{fig:app_cps_eta_sweep}
\end{figure*}

Matching the mean utility range does not eliminate the gap.
At $\eta=1.43$, CPS retains an allocation-error rate of
$17.2\%$ and a cutoff-crossing rate of $14.3\%$, compared
with $11.2\%$ and $8.3\%$ under APSM. Its NDCG@5 is
$0.304$, compared with $0.356$ under APSM.

CPS becomes statistically comparable on aggregate metrics only
near $\eta=8$. This comparison does not imply mechanism
equivalence: increasing $\eta$ eightfold also increases the
proper-score transfer scale eightfold and its audit-induced
variance component by a factor of $64$ under inverse-propensity
settlement.

The distinction is structural near a Top-$K$ boundary. For a
cutoff gap $\delta$, CPS penalizes the minimum crossing
deviation by $\eta\delta^2$, whereas APSM additionally incurs
a first-order belief-space allocation loss proportional to
$d_K\delta$. Consequently, reproducing the same local
deterrence using CPS alone requires $\eta$ to grow as the cutoff
gap decreases.

\subsection{Controlled Best-Response Implementation Details}
\label{app:best_response_details}

\paragraph{Allocation and utility parameters.}
All AgentWebBench experiments use ranked Top-$K$ allocation
with $K=2$ and position discounts
\begin{equation}
(d_1,d_2)=(1,0.63093),
\end{equation}
where $0.63093=1/\log_2 3$. The platform-controlled private
selection benefit is
\begin{equation}
G_i(a)=v\,\mathbb I\{i\in a\},\qquad v=1.
\end{equation}
Thus, selection provides unit private utility regardless of
rank. Position-dependent value enters separately through the
audited realized-value and Groves terms. We use the normalized
Brier score
\begin{equation}
S(r,y)=1-(r-y)^2
\end{equation}
with proper-score weight $\eta=1$.

\paragraph{Admissible reports.}
For the exact calculator experiments, the agent-specific
candidate set is
\begin{equation}
\begin{split}
\mathcal R_i
={}&
\{0.01,0.25,0.50,0.75,0.99,q_i\}\\
&\cup
\left\{
[r_j-0.01]_{0.01}^{0.99},
[r_j+0.01]_{0.01}^{0.99}
:j\neq i
\right\},
\end{split}
\label{eq:calculator_report_set}
\end{equation}
where $[x]_{0.01}^{0.99}$ clips $x$ to
$[0.01,0.99]$. Duplicate values are removed and the remaining
reports are sorted. Consequently, the frozen reference belief
$q_i$ is always admissible, while the additional points capture
reports immediately below and above competitors' allocation
thresholds.

The finite-grid ablation instead uses
\begin{equation}
\mathcal R^{\mathrm{grid}}
=
\{0.05,0.10,\ldots,0.95\}.
\end{equation}
In this condition, $q_i$ need not lie on the grid; the resulting
difference therefore measures report-discretization error.

\paragraph{Tie-breaking.}
Allocation ranks agents by decreasing report, with smaller
agent index breaking exact report ties. Utility ties are resolved
by selecting the smallest report among maximizing candidates.
The approximate-execution sweep selects the highest report among
numerically tied utility maximizers.
For Pivot-only, we retain and evaluate the complete maximizing
set rather than forcing a unique report.

\paragraph{Best-response protocol per experiment.}
Table~\ref{tab:best_response_protocols} summarizes the protocol
used in each experiment. Across all controlled-replay experiments
(RQ2 and RQ3), each agent optimizes independently against
$r_{-i}=q_{-i}$ (unilateral, simultaneous best response). After
all agents submit, the joint profile is evaluated under the
common Top-$K$ allocation rule. All mechanisms within an
experiment share the same protocol.

\begin{table}[h]
\centering
\caption{Best-response protocol per experiment.}
\label{tab:best_response_protocols}
\footnotesize
\begin{tabular}{@{}p{1.5cm}p{1.5cm}p{2.5cm}p{1.3cm}@{}}
\toprule
Experiment & $r_{-i}$ & $\mathcal R_i$ & Optimization \\
\midrule
RQ2 AgentWebBench & $q_{-i}$ (unilateral) & cached, Eq.~\eqref{eq:calculator_report_set} & exact enum. \\
RQ2 MultiAgentBench & $q_{-i}$ (unilateral) & cached domain & exact enum. \\
RQ3 exact & $q_{-i}$ (unilateral) & $q_i$-augmented, Eq.~\eqref{eq:calculator_report_set} & exact enum. \\
RQ3 approximate & $q_{-i}$ (unilateral) & stated domain & noisy / MC \\
RQ3 near-threshold & synthetic fixed (see text) & stated domain & exact enum. \\
\bottomrule
\end{tabular}
\end{table}

For the near-threshold decomposition, we generate 50 controlled
instances according to
\begin{equation}
q_i\sim\mathrm{Unif}(0.075,0.925),
\end{equation}
and
\begin{equation}
\tau_i=
\operatorname{clip}\!\left(
q_i+\mathrm{Unif}(-0.035,0.035),
0.05,0.95
\right).
\end{equation}
The competing profile is fixed as
\begin{equation}
r_{-i}=(0.9,\tau_i,0.1).
\end{equation}
This construction places the evaluated agent near a hard
allocation boundary without introducing best-response dynamics.

\paragraph{Approximate utility evaluation.}
For additive-noise execution, each candidate utility is
perturbed as
\begin{equation}
\widetilde U_i(r)
=
\overline U_i(r)+\xi_{i,r},
\qquad
\xi_{i,r}\sim\mathcal N(0,\sigma^2),
\end{equation}
with $\sigma\in\{0.01,0.05\}$. CPS and APSM reuse the same
noise realization for each paired agent--report instance.

For Monte Carlo execution, we draw
\begin{equation}
Y^{(m)}\sim\operatorname{Bernoulli}(q_i),
\qquad
M\in\{10,50,200\},
\end{equation}
and replace the expected proper score by its sample average.
Only the outcome $Y$ is sampled. The audit indicator is
integrated out analytically, while selection compensation,
expected verified value, pivot charges, and inverse-propensity
scaling remain part of the analytically evaluated expected
utility.

\end{document}